\documentclass[manuscript,screen]{acmart}
\usepackage{graphicx}
\usepackage{subcaption}
\usepackage{booktabs}
\usepackage{siunitx}
\usepackage{caption}
\usepackage{amsmath}
\usepackage{mathtools}
\usepackage{threeparttable}
\usepackage{enumitem}
\usepackage{makecell}
\usepackage{fancyvrb}
\usepackage{multirow}
\usepackage{ebproof}
\usepackage{listings}
\usepackage[most]{tcolorbox} 
\usepackage{wrapfig} 
\usepackage{xparse}
\usepackage{adjustbox} 
\usepackage[linesnumbered,ruled,vlined]{algorithm2e} 
\usepackage{tikz}
\usetikzlibrary{fit,shapes,backgrounds}

\allowdisplaybreaks[3]
\newcommand{\stlctype}[3]{{#1} \vdash {#2} : {#3}}
\newcommand{\sketch}[1]{#1^{\circ}} 
\newcommand{\meta}[1]{
  \ensuremath{%
    #1
  }
}

\newcommand{\acquire}[1]{\mathrm{Acquire}(#1)}
\NewDocumentCommand{\unify}{ m }{\!\mathrm{Unify}(\smash{#1})}
\newcommand{\ensure}[1]{#1} 
\definecolor{skunknown}{RGB}{180,0,180}
\newcommand{\sk}[1]{\textcolor{skunknown}{#1}} 
\newcommand{\skhzc}[1]{#1}
\newcommand{\mvhzc}{\text{FV}}
\newcommand{\welltyped}[1]{\mathit{well\text{-}typed}\left(#1\right)}
\newcommand{\func}[1]{\mathrm{#1}}          
\newcommand{\updir}{\textcolor{black}{↑}}  
\newcommand{\downdir}{\textcolor{black}{↓}}

\definecolor{rycolor}{RGB}{150, 0, 0}

\newcommand{\toolname}[1]{\textsc{#1}} 
\newcommand{\tcode}[1]{{\small \texttt{\spaceskip=0.2em #1}}}
\newcommand{\scode}[1]{{\fontsize{8.2pt}{\baselineskip}\tt \texttt{\spaceskip=0.15em #1}}}
\newcommand{\redbox}[1]{{\setlength{\fboxsep}{2pt}\colorbox{red!20}{#1}}}
\newcommand{\greenbox}[1]{{\setlength{\fboxsep}{2pt}\colorbox{green!20}{#1}}}
\newcommand{\mainname}{\toolname{TyFlow}\xspace}
\newcommand{\stlc}{\lambda^{\rightarrow}}

\AtBeginDocument{}

\acmJournal{TOSEM}
\acmVolume{0}
\acmNumber{0}
\acmArticle{0}
\acmYear{2026}
\acmMonth{1}

\setcopyright{none}

\begin{document}
\title{\mainname: A Type-Aware Approach to Neural Code Models}
\author{Zhechong Huang}
\email{willhuang@stu.pku.edu.cn}
\affiliation{%
  \institution{Peking University}
  \city{Beijing}
  \country{China}
}
\author{Zhao Zhang}
\email{zhangzhao2019@pku.edu.cn}
\affiliation{%
  \institution{Peking University}
  \city{Beijing}
  \country{China}
}
\author{Ruyi Ji}
\email{jiry@umich.edu}
\affiliation{%
  \institution{University of Michigan}
  \city{Ann Arbor}
  \country{USA}
}
\author{Tingxuan Xia}
\email{kkogoro@stu.pku.edu.cn}
\affiliation{%
  \institution{Peking University}
  \city{Beijing}
  \country{China}
}
\author{Qihao Zhu}
\email{Zhuqh@pku.edu.cn}
\affiliation{%
  \institution{Peking University}
  \city{Beijing}
  \country{China}
}
\author{Qinxiang Cao}
\email{caoqinxiang@gmail.com}
\affiliation{%
  \institution{Shanghai Jiao Tong University}
  \city{Shanghai}
  \country{China}
}
\author{Zeyu Sun}
\email{zeyu.zys@gmail.com}
\affiliation{%
  \institution{Institute of Software, Chinese Academy of Sciences}
  \city{Beijing}
  \country{China}
}
\author{Wiggin Zhou}
\email{wigginzhou@tencent.com}
\affiliation{%
  \institution{Tencent}
  \city{Beijing}
  \country{China}
}
\author{Yingfei Xiong}
\authornote{Corresponding author.}
\email{xiongyf@pku.edu.cn}
\affiliation{%
  \institution{Peking University}
  \city{Beijing}
  \country{China}
}


\begin{abstract}
Language models have shown remarkable proficiency in code generation; nevertheless, ensuring type correctness remains a challenge. Although traditional methods, such as constrained decoding, alleviate this problem by externally rejecting untypable code, the model itself does not effectively learn type reasoning internally, which ultimately limits its overall performance. 

This paper introduces \mainname, a novel system that internalizes type reasoning within code generation to guide the model to learn the type system. 
The core of our approach is a novel type-guided program synthesis system that maintains an isomorphism between type derivation trees and synthesis derivation trees, enabling a new code representation based on synthesis decision sequences rather than traditional text-based token sequences. By offloading the complexity of type system learning to the representation itself, models can redirect their computational resources toward higher-level program semantics.

Our evaluation shows that \mainname not only eliminates type errors but also significantly improves functional correctness, highlighting the importance of aligning LMs with type systems internally.

\end{abstract}




\maketitle

\section{Introduction} \label{section:intro}
Recently, language models (LMs) have achieved great success in code generation~\cite{li2022competition, DBLP:conf/pldi/Xu0NH22, zheng2023codegeex, DBLP:journals/corr/abs-2108-07732}, where code is generated based on a natural language prompt.
LMs typically treat programs as plain text, encoding them as sequences of tokens and generating code by iteratively predicting the next token given the current partial program. After being trained on large corpora of code and natural language, such LMs can generate programs on demand with high accuracy.

However, language models (LMs) frequently struggle with type systems, often producing code that violates basic type constraints. 
Empirical studies indicate that type errors alone account for 33.6\% of failed LM-generated programs~\cite{DBLP:journals/ese/TambonDNKDA25, DBLP:journals/corr/abs-2407-06153}. 
Notably, even cutting-edge tools powered by state-of-the-art LMs can generate code that fails to compile.
For instance, an evaluation of GitHub Copilot on LeetCode revealed that 24\% of its suggestions resulted in compilation errors~\cite{DBLP:conf/msr/NguyenN22}, which mainly resulted from type errors~\cite{DBLP:journals/corr/abs-2504-09246}. 
This type validity problem becomes particularly pronounced in resource-limited languages or those with complex type systems~\cite{DBLP:journals/corr/abs-2406-00515,DBLP:journals/corr/abs-2410-03981}.

Ensuring type correctness is also useful beyond reducing type errors. First, the type system is essentially a static analysis system, and in theory, we can incorporate any decidable constraint into the type system.
For example, Rust ensured no invalid memory access through an ownership type system. Therefore, by enriching the type system, we can enforce any decidable safety property of the code. Second, as revealed by recent studies~\cite{zhang2025gramtransbettercoderepresentation, DBLP:conf/acl/LiangZ0LLXCZZZC25}, even if LMs can learn to satisfy simple code constraints such as syntactic correctness given sufficient training data, by assisting LMs through external facilities for satisfying the constraints during training and inference, the overall performance of the model will further improve, possibly because the computational resources of the model can then be saved for higher-level tasks, such as algorithm design and aligning the semantics between the programs and the natural language description. 

Ensuring type correctness in practice, however, remains challenging. 
As illustrated in \autoref{fig:codet5-incorrect}, which presents an ill-typed output from a widely-used and fine-tuned model \toolname{CodeT5}~\cite{DBLP:conf/emnlp/0034WJH21}, the model incorrectly references an undefined variable \tcode{num}, which could have been prevented through proper type checking.
Essentially, the challenge originates from the \textit{huge gap} between the text-based token sequence representation of programs and the actual structural process of type checking.

\begin{figure}
\hfill
    \begin{subfigure}{0.45\textwidth}
        \begin{lstlisting}[language=java, basicstyle=\footnotesize\ttfamily, escapeinside={@}{@}, numbers=left, numbersep=2pt]
int firstOdd(List<Integer> l) {
  for (int i = 0; i < l.size(); i++)
    if (l.get(i) % 2 != 0)
      return @\greenbox{l.get(i)}@;
  return -1;
}
        \end{lstlisting}
\vspace{-0.5em}
\caption{The expected function.}
\vspace{1em}
        \label{fig:codet5-correct}
    \end{subfigure}
    \hfill
    \begin{subfigure}{0.45\textwidth}
\centering 
\begin{minipage}{0.6\textwidth}
\begin{lstlisting}[language=java, basicstyle=\footnotesize\ttfamily, escapeinside={@}{@}, numbers=left, numbersep=3pt, firstnumber=3]
if (l.get(i) % 2 != 0)
  return @\redbox{num}@;
\end{lstlisting}
\end{minipage}
\vspace{-0.5em}
\caption{The ill-typed code generated by \toolname{CodeT5}}
\label{fig:codet5-incorrect}
\vspace{1em}
\begin{minipage}{0.6\textwidth}
\begin{lstlisting}[language=java, basicstyle=\footnotesize\ttfamily, escapeinside={@}{@}, numbers=left, numbersep=3pt, firstnumber=3]
if (l.get(i) % 2 != 0)
  return @\redbox{i}@;
\end{lstlisting}
\end{minipage}
\vspace{-0.5em}
\caption{The still-incorrect code after incorporating constraint decoding for types.}
        \label{fig:codet5-constraint-incorrect}
    \end{subfigure}
\hfill\ 
\caption{A code generation task and the incorrect Java programs generated by \toolname{CodeT5}, where the prompt comprises the function signature, input-output examples, and the description ``\emph{return the first odd element from a list (-1 by default)}''. We highlight the incorrect code as \redbox{red} and the expected counterpart as \greenbox{green}.}
    \label{fig:intro-sample}
\end{figure}

\begin{itemize}
    \item During training, LMs can only see sequences of syntax tokens, which provides little information on the underlying type system. From these sequences alone, LMs struggle to recover and learn the entire type system.
    For example, typing the function in \autoref{fig:codet5-correct} relies on the signature of the \tcode{get} method and the auto-unboxing rule in Java that allows \tcode{Integer} to be returned as \tcode{int}, but both rules are obscured within the text-based token sequence.
    
    \item During generation, LMs must ensure the type correctness of the produced code. However, type-checking the partially generated code is challenging, as it often requires complex analysis to gather necessary contextual information from the existing code. 
    For instance, to check variable availability, the LM must identify all prior definitions (such as \tcode{l} and \tcode{i} in \autoref{fig:codet5-correct}) and determine whether each variable remains in scope, which necessitates non-trivial global reasoning across the entire codebase. Existing LMs often struggle with such complex contextual reasoning, especially as the codebase grows in size~\cite{guo_2025,DBLP:conf/acl/BaiTZ0WLCX0D0L25}.
\end{itemize}

To enforce type correctness, a typical approach is \emph{rejecting invalid generations}, as realized by a series of {constrained decoding} approaches~\cite{DBLP:conf/sigsoft/0003X023,DBLP:journals/corr/abs-2504-09246}.
These approaches leverage external type checkers and reject untypable programs generated by LMs.
However, such an external patch cannot bridge the gap in principle.
Constrained decoding can only exclude ill-typed programs; it does not help the model learn type systems, and thus cannot improve, and may even worsen, the distribution among well-typed programs.~\cite{DBLP:conf/nips/ParkWBPD24}.
For example, in our sample task, \toolname{CodeT5} assigns a low probability to the expected expression (\autoref{fig:codet5-correct}), possibly because the LM is unaware of the auto-unboxing rule and thus incorrectly identifies it as ill-typed.
Constrained decoding helps little with this case, as shown in \autoref{fig:codet5-constraint-incorrect} -- it merely shifts to a still incorrect but well-typed program that avoids auto-unboxing.

To guide the model in learning the type system, a straightforward approach is to employ a reasoning model that augments the program with a type correctness proof; that is, the model first generates the program and then generates its type correctness proof as the type reasoning process, or vice versa. In this way, the type reasoning is made explicit to the model, enabling it to learn the type system of the target program.
However, the separation between program generation and type reasoning introduces fundamental limitations. During training, presenting them as separate phases prevents LMs from capturing their fine-grained interactions, requiring extra effort to align them. During inference, the model must still perform global type reasoning without localized typing context, leaving the challenge of contextual reasoning unresolved.
As a result, this approach not only hinders learning efficiency but also causes inconsistencies between type reasoning and generated code. As our evaluation will later demonstrate, it does not yield satisfactory performance.

Building on these insights, we aim to propose a new synthesis system that \emph{reconciles type systems internally with LMs}. 
Our approach leverages the fact that LMs are not confined to generating programs as plain text; they can be incorporated into any search-based system that constructs programs via sequences of decisions, with each decision made in a specific context. 
For any such system, an LM can be trained to predict the optimal decision at each step, enabling program synthesis by interacting with the LM throughout the search process.

In this way, reconciling type systems with LMs reduces to designing a synthesis system $\mathcal S$ that cooperates effectively with LMs while ensuring type correctness.
Specifically, we expect the system $\mathcal S$ to satisfy the following properties, which succinctly summarize the main challenges of LM-based code generation as we discussed earlier:
\begin{itemize}
 \item \textbf{Type Explicitness}. $\mathcal{S}$ should explicitly trace the type derivation throughout the decision process, making the type reasoning visible to the model.
 \item \textbf{Context Locality}. At each step, $\mathcal{S}$ should present all necessary type information within a small, bounded context window. This relieves LMs from the burden of inferring such information from lengthy contexts.
 \item \textbf{Derivation Vicinality}. The decisions in $\mathcal{S}$ should be represented such that each program fragment appears adjacent to its type derivation, making it easier for LMs to learn the correspondence between code and type reasoning during training.
 \item \textbf{Data Usability}. The decision sequences of $\mathcal{S}$ and source code should be automatically and bidirectionally convertible, enabling both the extraction of decision sequences from existing programs and the reconstruction of programs from generated decision sequences.
\end{itemize}

\textbf{\textit{Our first contribution}} is the observation that, under constructive logic, an existential proof of type correctness inherently contains the program it proves. Formally, an existential proof $\Pi$ of the following statement
\[
\exists p.\; \mathsf{welltyped}(p)
\]
must explicitly construct a witness program $p$. Thus, there exists an extraction function $\mathsf{extract}$ such that $\mathsf{extract}(\Pi) = p$. We leverage this insight to design a novel synthesis system and program representation that unifies program generation with type reasoning.

Based on this observation, the process of constructing an existential proof can be reformulated as a program synthesis process. The proof is built by applying typing rules step by step, where each step makes a \textit{decision} that either applies a rule to decompose the current goal into subgoals, or instantiates a variable within the context. Since the program is embedded within the proof, this sequence of decisions simultaneously determines the program being synthesized.

We use this decision sequence as a novel program representation, which satisfies the desired properties:
\begin{itemize}
    \item \textbf{Type Explicitness}. Constructing the decision sequence is equivalent to constructing the existential type correctness proof, making the complete type reasoning explicit throughout the process.
    \item \textbf{Derivation Vicinality}. The representation inherently captures both the program and its type correctness proof simultaneously, as they arise from the same construction process.
\end{itemize}

Notably, all proof rules in our system are formulated as constrained Horn clauses (see \autoref{chc-intro}). This representation not only subsumes the prior work on syntactic correctness~\cite{DBLP:conf/icse/ZhuL000C24,DBLP:conf/acl/LiangZ0LLXCZZZC25} but also accommodates a broader class of Turing-computable specifications. While this paper focuses on type constraints, our synthesis method is general and applicable to arbitrary logical constraints.

\textbf{\textit{Our second contribution}} is a novel encoder-decoder model architecture tailored for our synthesis framework. During synthesis, the current proof state evolves at each step, and the model needs to extract the dynamic typing context from this state to make the next decision. Our architecture addresses this by using the encoder to process both the natural language specification and the current synthesis goal, while the decoder autoregressively generates the decision sequence. This design enables the model to access localized type information at each step without requiring global reasoning over the entire codebase, thereby satisfying the \textbf{Context Locality} property.

\textbf{\textit{Our third contribution}} is \mainname, an automated system for training LMs to generate well-typed programs. \mainname takes as input (1) the language definition with its syntax, typing rules, and associated toolchain (including a parser and a type checker), and (2) a set of training tasks comprising prompts and target programs. 
From these inputs, \mainname automatically instantiates the synthesis framework for the target language, providing: (1) a training component that extracts decision sequences from programs and their existential type correctness proofs, and (2) a generation component that first generates decision sequences within the synthesis system, and then extracts well-typed programs from them.
This satisfies the \textbf{Data Usability} property: any well-typed program permits the extraction of its corresponding decision sequence, while any valid decision sequence can be reconstructed into a well-typed program along with its proof.

\smallskip
Our experimental results show that, relative to vanilla code generation, \mainname not only eliminates untypable programs\footnote{with respect to the type system we implemented}, but also substantially increases the number of programs that pass all unit tests. These findings indicate that \mainname successfully helps the LM acquire a deeper understanding of the type system and yields a more precise distribution over correct programs.

The rest of the paper is organized as follows. 
\autoref{sec:overview} gives an overview of our approach. 
\autoref{sec:meta} introduces how we translate the typing rules into synthesis rules. 
\autoref{sec:approach} proves the correctness of our framework and introduces how we prepare the training data and perform synthesis within the framework. 
\autoref{sec:implementation} introduces our neural network architecture. 
\autoref{sec:evaluation} describes the evaluation of our approach. 
Finally, \autoref{sec:related} discusses related works and \autoref{sec:conclusion} concludes the paper. 
\section{Overview \label{sec:overview}}

This section gives an overview of our approach \mainname, illustrating how it generates a well-typed program in the simply typed lambda calculus $\stlc$(a foundational functional programming language)~\cite{DBLP:journals/jsyml/Church40}.

\subsection{Problem Definition} \label{subsection:moti-problem}
We formulate the type-safe code generation problem based on two primary inputs: a language specification and a dataset of training tasks.

\newcommand{\checkBinding}[3]{\ensuremath{{#2}:{#3} \in {#1}}}
\newcommand{\abs}[3]{\ensuremath{\lambda {#1}:{#2}.\ {#3}}}
\newcommand{\vari}[1]{\ensuremath{{#1}}}
\newcommand{\snoc}[3]{\ensuremath{{#1}, {#2}:{#3}}}
\newcommand{\arr}[2]{\ensuremath{{#1}\rightarrow {#2}}}
\newcommand{\app}[2]{\ensuremath{{#1}\ {#2}}}
\newcommand{\rulename}[1]{\textsc{#1}}

\begin{figure*}[t]
\hfill
\begin{subfigure}{0.65\textwidth}
\begin{lstlisting}[basicstyle=\fontsize{8.2pt}{\baselineskip}\tt, escapeinside={@}{@}, morekeywords={data, pred, rule, func}, basewidth={0.5em, 0.2em}] 
data Prog = true | false | Var | Prog Prog | @$\tt \lambda$@Str:Type. Prog
data Type = bool | Type @$\rightarrow$@ Type      
data Context = empty | Context, Var:Type
\end{lstlisting}
\vspace{-0.5em}
\caption{The inductive data types for the syntax.} \label{figure:stlc-syntax}
\end{subfigure}
\hfill\ 

\vspace{1em}

\begin{subfigure}{0.8\textwidth}
\small

\begin{gather*}
\frac{\checkBinding{\Gamma}{x}{t}}{\stlctype{\Gamma}{\vari x}{t}} \ \ \textsc{(T-Var)} \qquad
\frac{\stlctype{\snoc{\Gamma}{x}{t_1}}{p}{t_2}}{\stlctype{\Gamma}{\left(\abs{x}{t_1}{p}\right)}{\arr{t_1}{ t_2}}} \ \ \textsc{(T-Abs)} \qquad
\frac{\stlctype{\Gamma}{p_1}{\arr{t_1}{t_2}} \quad \stlctype{\Gamma}{p_2}{t_1}}{\stlctype{\Gamma}{\app{p_1}{p_2}}{t_2}} \ \ \textsc{(T-App)} \qquad
\end{gather*}
\vspace{-0.5em}
\caption{Selected typing rules for $\stlc$. Each typing rule derives a typing judgment of the form $\stlctype{\Gamma}{p}{t}$, which asserts that program $p$ has type $t$ under context $\Gamma$} 
    \label{figure:stlc-rules}
\end{subfigure}
\vspace{-0.5em}
\caption{The definition of $\stlc$ in \mainname, comprising its syntax and some selected typing rules..
}
\label{figure:stlc-definition}
\end{figure*}

\smallskip 
\noindent \textbf{Language definition}. As shown in \autoref{figure:stlc-definition}, a language in \mainname is specified by its syntax and its type system.
In this paper, we define the syntax as a set of inductive data types over several built-in base types (e.g., \tcode{string}, \tcode{bool}). 
Here, the syntax of $\stlc$ consists of three data types (\autoref{figure:stlc-syntax}):
\tcode{Prog} and \tcode{Type} define programs and types in $\stlc$ as the standard simply typed lambda calculus, above the base type \tcode{bool}; and 
\tcode{Context} specifies a context for variable types, where \tcode{empty} denotes the empty context, and ${(\snoc{\Gamma}{\tt x}{\tt t})}$ extends a context $\Gamma$ by binding the type \tcode{t} to the variable \tcode{x}.

To specify the type system, we require the definition of each typing rule (which defines how types are assigned to program constructs) together with an executable type checker.
For clarity, we present our approach using the standard type notations of $\stlc$ as listed in \autoref{figure:stlc-rules}.
The type checker is expected to generate a type correctness proof (i.e., a type derivation tree, where each tree node corresponds to an application of a typing rule, like \autoref{figure:sample-type-tree}) for the typing judgement $\stlctype{\tt empty}{p}{t}$, which asserts that program $p$ has type $t$ under the empty context.

\smallskip 
\noindent \textbf{Training tasks}. Following the standard code generation paradigm, the objective is to generate correct programs from natural language specifications. Therefore, each training task in \mainname is specified by a prompt paired with the expected program.
\autoref{figure:sample-task} introduces an example task over $\stlc$, and \autoref{figure:sample-type-tree} illustrates the type correctness proof of the expected program generated by the type checker during type checking.

\begin{figure*}[t]
\begin{minipage}{0.3\textwidth}
    \small 
    \fbox{
        \begin{minipage}{0.98\textwidth}
\textbf{Prompt}: Implement an identity function that takes a boolean input and returns the same value; then apply it to the value true.
 
\vspace{0.5em}
\textbf{Expected program} $p^*$: 
\begin{align*}
    \app{
        (\abs{\scode{x}}{\scode{bool}}{\vari{\scode{x}}})
    }{
        \text{\scode{true}}
    }
\end{align*}
\end{minipage}
}
\caption{A training task over $\stlc$.}
\label{figure:sample-task}
\end{minipage}
\hfill
\begin{minipage}{0.65\textwidth}
\small 
\hfill\ 

\begin{prooftree}[rule margin=1ex, separation=1em]
\hypo{
  \checkBinding{
    \left\{\snoc{\scode{empty}}{\scode{x}}{\scode{bool}}\right\}}
    {\scode{x}}
    {\scode{bool}}
}
\infer1[(\textsc{T-Var})]{\stlctype{\snoc{\scode{empty}}{\scode{x}}{\scode{bool}}}{\scode{x}}{\scode{bool}}}
\infer1[(\textsc{T-Abs})]{\stlctype{\scode{empty}}{\abs{\scode{x}}{\scode{bool}}{\scode{x}}}{\arr{\scode{bool}}{\scode{bool}}}}
\hypo{~} 
\infer1[(\textsc{...})]{\stlctype{\scode{empty}}{\scode{true}}{\scode{bool}}}
\infer2[(\textsc{T-App})]{\stlctype{\scode{empty}}{\app{(\abs{\scode{x}}{\scode{bool}}{\scode{x}})}{\scode{true}}}{\scode{bool}}}
\end{prooftree}

\caption{The type correctness proof of the expected program $p^*$ in \autoref{figure:sample-task}.}
\label{figure:sample-type-tree}
\end{minipage}
\end{figure*}

\subsection{From Proof Construction to Program Synthesis} \label{subsection:moti-representation}
Given the inputs in \autoref{figure:sample-task}, \mainname aims to generate well-typed $\stlc$ programs from prompts. Our core insight is that we can construct an existential type correctness proof even when the program is unknown, thereby synthesizing the program while constructing the proof. This is possible because each typing rule restricts a part of the program structure, making the structure of the proof sufficient to determine the entire program.

For example, in \autoref{figure:sample-type-tree}, the \rulename{T-App} rule restricts the top-level operator of the program to be a function application, the \rulename{T-Abs} rule restricts the function term to be a lambda abstraction, and the application of \rulename{T-Var} defines that the body of the lambda abstraction is a variable. 
In other words, \textbf{typing rules inherently serve a program synthesis role}. This observation allows us to construct a synthesis system directly grounded in typing rules, enabling simultaneous program synthesis and type reasoning.

We illustrate the synthesis process of \mainname in detail as shown in \autoref{tab:synthesis-process}, where we use \sk{purple} to distinguish the unknown variables (i.e., placeholders to be filled in) in each goal, which will be instantiated later during the synthesis.
For the $\stlc$ language, \mainname begins by specifying the synthesis goal as $\stlctype{\tcode{empty}}{\sk{p}}{\sk{t}}$. This synthesis goal is a typing judgment with unknown variables, and it corresponds to constructing an existential proof of $\exists \sk{p}, \sk{t}.\; \stlctype{\tcode{empty}}{\sk{p}}{\sk{t}}$: the unknown variables in the typing judgment are exactly the existential variables to be witnessed in the proof. The initial context $\tcode{empty}$ indicates that no variables have been bound at the beginning.

During synthesis, the system incrementally applies typing rules to construct the proof, and fills in the values of $\sk{p}$ and $\sk{t}$ correspondingly. For other programming languages with different type systems, the synthesis goal can be adapted accordingly. 

\begin{table}[t]
\centering
\caption{Synthesis Process of \mainname for program in \autoref{figure:sample-task}}
\label{tab:synthesis-process}
\resizebox{\textwidth}{!}{%
\begin{tabular}{|c|l|l|l|l|l|}
\hline
\textbf{Step} & \textbf{Current Goal} & \textbf{Rule Applied} & \textbf{Substitution} & \textbf{New Subgoals} & \textbf{Current Program} \\
\hline
\multirow{2}{*}{\textbf{1}} & \multirow{2}{*}{$\stlctype{\tcode{empty}}{\sk{p}}{\sk{t}}$} & \multirow{2}{*}{\rulename{T-App}} & $\sigma_1=\{ \sk{\Gamma_1} \mapsto \tcode{empty},$ & $\stlctype{\tcode{empty}}{\sk{p_1}}{\arr{\sk{t_1}}{\sk{t_2}}}$ & \multirow{2}{*}{$\app{\sk{p_1}}{\sk{p_2}}$} \\
& & & $\sk{p} \mapsto \app{\sk{p_1}}{\sk{p_2}}, \sk{t} \mapsto \sk{t_2}\}$ & $\stlctype{\tcode{empty}}{\sk{p_2}}{\sk{t_1}}$ & \\
\hline
\multirow{2}{*}{\textbf{2}} & \multirow{2}{*}{$\stlctype{\tcode{empty}}{\sk{p_1}}{\arr{\sk{t_1}}{\sk{t_2}}}$} & \multirow{2}{*}{\rulename{T-Abs}} & $\sigma_2 = \{ \sk{\Gamma_2} \mapsto \tcode{empty},$ & \multirow{2}{*}{$\stlctype{\snoc{\tcode{empty}}{\sk{x_1}}{\sk{t_3}}}{\sk{p_3}}{\sk{t_4}}$} & \multirow{2}{*}{$\app{(\abs{\sk{x_1}}{\sk{t_3}}{\sk{p_3}})}{\sk{p_2}}$} \\
& & & $\sk{p_1} \mapsto \abs{\sk{x_1}}{\sk{t_3}}{\sk{p_3}}, \ldots\}$ & & \\
\hline
\multirow{3}{*}{\textbf{3}} & \multirow{3}{*}{$\stlctype{\snoc{\tcode{empty}}{\sk{x_1}}{\sk{t_3}}}{\sk{p_3}}{\sk{t_4}}$} & \multirow{3}{*}{\rulename{T-Var}} & $\sigma_3=\{ \sk{\Gamma_3} \mapsto (\snoc{\tcode{empty}}{\sk{x_1}}{\sk{t_3}}),$ & \multirow{3}{*}{None} & \multirow{3}{*}{$\app{(\abs{\tcode{x}}{\tcode{bool}}{\tcode{x}})}{\sk{p_2}}$} \\
& & & $\sk{p_3} \mapsto \vari{\sk{x_2}}, \sk{t_4} \mapsto \sk{t_5}\}$ & & \\
& & & $\sigma_4 = \{ \sk{x_1} \mapsto \tcode{x}, \sk{x_2} \mapsto \tcode{x}, \ldots\}$ & & \\
\hline
\multirow{2}{*}{\textbf{4}} & \multirow{2}{*}{$\stlctype{\tcode{empty}}{\sk{p_2}}{\tcode{bool}}$} & \multirow{2}{*}{\rulename{...}} & \multirow{2}{*}{$\sigma_6=\{\sk{p_2}\mapsto\tcode{true}, \ldots\}$} & \multirow{2}{*}{None } & \multirow{2}{*}{$\app{(\abs{\tcode{x}}{\tcode{bool}}{\tcode{x}})}{\tcode{true}}$} \\
& & & & & \\
\hline
\end{tabular}%
}
\end{table}

\smallskip \noindent \textbf{Step 1: applying \rulename{T-App}.}
Given the initial synthesis goal $\stlctype{\tcode{empty}}{\sk{p}}{\sk{t}}$, the LM must select a typing rule to apply. Assume the rule \rulename{T-App} is selected. The system then verifies that this application is admissible for the current goal, applies the rule, and generates the corresponding subgoals (one for each premise). More concretely, the system performs the following steps in order:
\begin{enumerate}
    \item First, \mainname $\alpha$-renames the rule by replacing all variables in the rule with fresh ones to avoid name conflicts. Under this renaming, \rulename{T-App} takes the following form.
    $$\frac{\stlctype{\sk{\Gamma_1}}{\sk{p_1}}{\arr{\sk{t_1}}{\sk{t_2}}} \quad \stlctype{\sk{\Gamma_1}}{\sk{p_2}}{\sk{t_1}}}{\stlctype{\sk{\Gamma_1}}{\app{\sk{p_1}}{\sk{p_2}}}{\sk{t_2}}} $$
    
    \item Subsequently, \mainname attempts to unify the rule's conclusion judgment with the current synthesis goal $\stlctype{\tcode{empty}}{\sk{p}}{\sk{t}}$, accumulating the induced constraints as the equations shown below. If unification fails, the rule application is considered invalid and discarded.
    $$
    \sk{\Gamma_1} = \tcode{empty} \ \ \wedge\ \ \app{\sk{p_1}}{\sk{p_2}} = \sk{p}  \ \ \wedge \ \  \sk{t_2} = \sk{t}
    $$

    \item \label{action3}Thereafter, \mainname attempts to resolve these constraints via unification, yielding a substitution that maps variables to their values. If unification fails, the rule application is rejected.
    $$\sigma_1=\big\{ \sk{\Gamma_1} \mapsto \tcode{empty}, \quad \sk{p} \mapsto \app{\sk{p_1}}{\sk{p_2}}, \quad \sk{t} \mapsto \sk{t_2}\big\}$$

    \item Finally, \mainname constructs synthesis subgoals by applying the substitution to each premise of \rulename{T-App}, which has two premises corresponding to the function and the argument of the application, 
    $\stlctype{\sk{\Gamma_1}}{\sk{p_1}}{\arr{\sk{t_1}}{\sk{t_2}}}$ and $\stlctype{\sk{\Gamma_1}}{\sk{p_2}}{\sk{t_1}}$, and \mainname will generate two subgoals:
    \begin{gather*}
    \stlctype{\tcode{empty}}{\sk{p_1}}{\arr{\sk{t_1}}{\sk{t_2}}} \quad \stlctype{\tcode{empty}}{\sk{p_2}}{\sk{t_1}}
    \end{gather*}
\end{enumerate}
Each subgoal corresponds to constructing an existential sub-proof. Specifically, the first subgoal $\stlctype{\tcode{empty}}{\sk{p_1}}{\arr{\sk{t_1}}{\sk{t_2}}}$ corresponds to proving $\exists \sk{p_1}, \sk{t_1}, \sk{t_2}.\; \stlctype{\tcode{empty}}{\sk{p_1}}{\arr{\sk{t_1}}{\sk{t_2}}}$, and the second subgoal corresponds to proving $\exists \sk{p_2}.\; \stlctype{\tcode{empty}}{\sk{p_2}}{\sk{t_1}}$. Once these sub-proofs are constructed with concrete witnesses, the substitution $\sigma_1$ determines the original existential variables: $\sk{p} = \app{\sk{p_1}}{\sk{p_2}}$ and $\sk{t} = \sk{t_2}$. \mainname proceeds to discharge the two subgoals recursively before closing the parent goal.

\smallskip \noindent \textbf{Step 2: applying \rulename{T-Abs}.}
\mainname applies the same procedure to the first subgoal $\stlctype{\tcode{empty}}{\sk{p_1}}{\arr{\sk{t_1}}{\sk{t_2}}}$. 
The LM will select the \rulename{T-Abs} with a high probability as this goal requires an arrow type.
\mainname then $\alpha$-renames the \rulename{T-Abs} rule:
$$
\frac{
    \stlctype{\snoc{\sk{\Gamma_2}}{\sk{x_1}}{\sk{t_3}}}{\sk{p_3}}{\sk{t_4}}
}{
    \stlctype{\sk{\Gamma_2}}{\left(\abs{\sk{x_1}}{\sk{t_3}}{\sk{p_3}}\right)}{\arr{\sk{t_3}}{\sk{t_4}}}
}
$$

By performing the remaining three actions as detailed in Step 1, \mainname again obtains a substitution and generates a new subgoal as follows.
$$\begin{array}{ll}
\textrm{Substitution: } & \sigma_2 = \big\{ \sk{\Gamma_2} \mapsto \tcode{empty}, \quad \sk{p_1} \mapsto \abs{\sk{x_1}}{\sk{t_3}}{\sk{p_3}}, \quad \sk{t_1} \mapsto \sk{t_3}, \quad \sk{t_2} \mapsto \sk{t_4}\big\}\\
\textrm{Subgoal: }& \stlctype{\snoc{\tcode{empty}}{\sk{x_1}}{\sk{t_3}}}{\sk{p_3}}{\sk{t_4}}
\end{array}$$

\smallskip \noindent \textbf{Step 3: applying \rulename{T-Var}.}
Next, \mainname handles the newly generated subgoal $\stlctype{\snoc{\tcode{empty}}{\sk{x_1}}{\sk{t_3}}}{\sk{p_3}}{\sk{t_4}}$, assuming the LM selects the \rulename{T-Var} rule. \mainname again $\alpha$-renames the rule and then obtains the substitution as shown below.
$$\begin{array}{ll}
\textrm{Rule: } & \displaystyle\frac{\checkBinding{\sk{\Gamma_3}}{\sk{x_2}}{\sk{t_5}}}{\stlctype{\sk{\Gamma_3}}{\vari{\sk{x_2}}}{\sk{t_5}}}\\[0.8em]
\textrm{Substitution: } & \sigma_3=\big\{ \sk{\Gamma_3} \mapsto (\snoc{\tcode{empty}}{\sk{x_1}}{\sk{t_3}}), \quad \sk{p_3} \mapsto \vari{\sk{x_2}}, \quad \sk{t_4} \mapsto \sk{t_5}\big\}
\end{array}$$

\rulename{T-Var} differs from previous rules in that its premise involves the membership predicate $\in$, which checks whether a variable binding exists in the context.
Rather than generating a new subgoal, \mainname directly invokes the predicate $\in$ to assess the validity of its arguments.
However, there are still unknown variables in both the premise and the synthesis goal even after applying the substitution $\sigma_3$:
$$\begin{array}{ll}
\textrm{Premise: } & \checkBinding{(\snoc{\tcode{empty}}{\sk{x_1}}{\sk{t_3}})}{\sk{x_2}}{\sk{t_5}}\\
\textrm{Goal: } & \stlctype{\snoc{\tcode{empty}}{\sk{x_1}}{\sk{t_3}}}{\sk{x_2}}{\sk{t_5}}
\end{array}$$

The unknown variables in the premise prevent \mainname from evaluating the predicate. 
To address this issue, \mainname will query the LM to provide an assignment to all such variables, such as follows:
\begin{gather}
\sigma_4 = \big\{ \sk{x_1} \mapsto \tcode{x},\quad \sk{x_2} \mapsto \tcode{x},\quad \sk{t_3} \mapsto \tcode{bool},\quad \sk{t_5} \mapsto \tcode{bool} \big\} \label{formula:sig4}
\end{gather}

By further applying $\sigma_4$, the premise becomes $\checkBinding{(\snoc{\tcode{empty}}{\tcode{x}}{\tcode{bool}})}{\tcode{x}}{\tcode{bool}}$, which is completely known; then, \mainname can evaluate the predicate and confirms that it is satisfied.
At this time, no new subgoal is generated, so the synthesis process for this branch is complete, where the result is the composition of $\sigma_3$ and $\sigma_4$ (i.e., applying both substitutions in sequence):
\[
\sigma_4 \sigma_3 = \big\{
\sk{\Gamma_3} \mapsto \snoc{\tcode{empty}}{\tcode{x}}{\tcode{bool}},\ 
\sk{p_3} \mapsto \tcode{x},\ 
\sk{t_4} \mapsto \tcode{bool},\ 
\sk{x_1} \mapsto \tcode{x},\ 
\sk{x_2} \mapsto \tcode{x},\ 
\sk{t_3} \mapsto \tcode{bool},\ 
\sk{t_5} \mapsto \tcode{bool}
\big\}
\]

\smallskip \noindent \textbf{Step 4: returning from subgoals.}
After finishing a subgoal, \mainname will apply the obtained substitution to all remaining subgoals to synchronize the progress; and once all subgoals are closed, \mainname will compose all obtained substitutions as the final substitution for the synthesis goal.

In this example, the substitution $\sigma_4\sigma_3$ obtained in Step 3 is further composed with the substitution $\sigma_2$ from Step 2, and the resulting composition $\sigma_5 = \sigma_4\sigma_3\sigma_2$ is returned as the result of the subgoal $\stlctype{\tcode{empty}}{\sk{p_1}}{\arr{\sk{t_1}}{\sk{t_2}}}$ in Step 2. This yields a concretized and verified typing judgment, $\stlctype{\tcode{empty}}{\abs{\tcode{x}}{\tcode{bool}}{\tcode{x}}}{\arr{\tcode{bool}}{\tcode{bool}}}$.
Since there are two subgoals generated in Step 1, \mainname applies the composed substitution $\sigma_5$ to the second subgoal, $\stlctype{\tcode{empty}}{\sk{p_2}}{\sk{t_1}}$, yielding
\[
\stlctype{\tcode{empty}}{\sk{p_2}}{\tcode{bool}}.
\]
\mainname then applies the same procedure to this subgoal, producing the substitution 
$\sigma_6=\{\sk{p_2}\mapsto\tcode{true}\ldots\}$. 
This new substitution is then composed with the previous $\sigma_1$ and $\sigma_5$ to yield $\sigma_7=\sigma_6\sigma_5\sigma_1$, which is applied to the original
synthesis goal $\stlctype{\tcode{empty}}{\sk{p}}{\sk{t}}$, giving the final result:
\[
\stlctype{\tcode{empty}}{\app{(\abs{\texttt{x}}{\tcode{bool}}{\texttt{x}})}{\tcode{true}}}{\tcode{bool}}
\]
At this point, we have successfully constructed an existential proof of $\exists \sk{p}, \sk{t}.\; \stlctype{\tcode{empty}}{\sk{p}}{\sk{t}}$, with the existential variables instantiated as $\sk{p} = \app{(\abs{\texttt{x}}{\tcode{bool}}{\texttt{x}})}{\tcode{true}}$ and $\sk{t} = \tcode{bool}$. The synthesized program $\sk{p}$ is thus guaranteed to be well-typed.

Performing program synthesis via existential proof construction embodies our two key principles: \textbf{Type Explicitness} and \textbf{Context Locality}.
\begin{itemize}
    \item \textbf{Type Explicitness.} Each synthesis step corresponds to applying a typing rule, which directly mirrors the structure of the type correctness proof. As shown in the example, the sequence of rule applications (\rulename{T-App}, \rulename{T-Abs}, \rulename{T-Var}, etc.) constructs a complete proof tree, making the type reasoning explicit throughout the synthesis process.
    \item \textbf{Context Locality.} At each step, the synthesis goal is a partially instantiated typing judgment that contains the current typing context. For instance, in Step 2, the goal $\stlctype{\tcode{empty}}{\sk{p_1}}{\arr{\sk{t_1}}{\sk{t_2}}}$ explicitly provides the context $\tcode{empty}$ and indicates that the program $\sk{p_1}$ to be synthesized must have an arrow type $\arr{\sk{t_1}}{\sk{t_2}}$. This local type information guides the LM to select the \rulename{T-Abs} rule, which produces a lambda abstraction.
\end{itemize}
\subsection{Translating Typing Rules to Synthesis Rules}
In summary of \autoref{subsection:moti-representation}, \mainname extends typing rules for synthesis with the two design choices:
\begin{itemize} 
    \item When applying a rule, instead of pattern matching, \mainname introduces rule parameters as fresh variables and solves the constraints between variables via unification.
    \item Whenever an executable predicate requires the value of an unknown variable, \mainname acquires this value directly from the LM.
\end{itemize}
With these ideas, typing rules can be restated as formal synthesis rules in a mechanical way, which will be elaborated in \autoref{sec:meta}. Here we shall illustrate this process with the \textsc{T-App} rule as in \autoref{figure:s-app}.

To synthesize programs in $\stlc$, a synthesis judgment takes the form:
\[
\stlctype{\sketch{\Gamma}}{\sketch{p}}{\sketch{t}} \leadsto \sigma
\]
\begin{itemize}
\item \label{sketch-label}The left side is an incomplete typing judgment, where $\sketch{\Gamma}$, $\sketch{p}$, and $\sketch{t}$ represent a (possibly) incomplete context, program, and type, respectively. For example, $\sketch{t}$ can be an unknown variable $\sk{t}$, a partially known type $\tcode{bool} \rightarrow \sk{t}$, or a fully known type $\tcode{bool} \rightarrow \tcode{bool}$.
\item The right side is an assignment $\sigma$ to all unknown variables on the left side.
\end{itemize}

This judgment means that \textbf{the left-hand side typing judgment can be properly completed by the right-hand side assignment}; in other words, after applying $\sigma$ as a substitution, $\stlctype{\sigma(\sketch{\Gamma})}{\sigma(\sketch{p})}{\sigma(\sketch{t})}$ will be a valid typing judgment in the type system of $\stlc$. 
From the perspective of proof construction, the synthesis judgment corresponds to constructing an existential proof: given a proof goal with unknown variables, the synthesis process finds witnesses (i.e., concrete values for these variables) and proves the validity of the goal.  

\begin{figure*}
\centering 
\includegraphics[width=\textwidth]{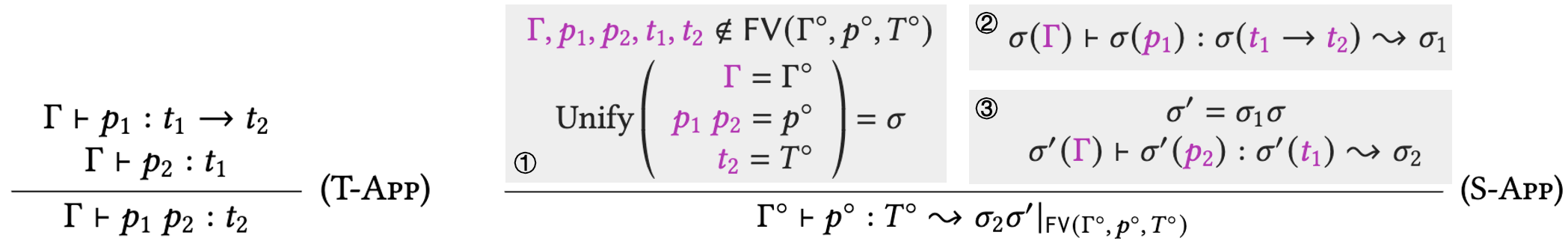}
\caption{The typing rule \textsc{T-App} and its corresponding synthesis rule \textsc{S-App}, where $\textsf{FV}(\sketch{\Gamma}, \sketch{p}, \sketch{T})$ denotes the unknown variables in the incomplete typing judgment. \textsc{S-App} will return assignments only for these variables.} \label{figure:s-app}
\end{figure*}

As marked in \autoref{figure:s-app}, the premises of rule \textsc{S-App} can be divided into three parts.
\begin{itemize}
    \item Part 1 introduces a fresh variable for each parameter in the typing rule \textsc{T-App}, unifies the conclusion of \textsc{T-App} with the current incomplete typing judgment (i.e., $\stlctype{\sk{\Gamma}}{\sk{p_1}~\sk{p_2}}{\sk{t}}$ with $\stlctype{\sketch{\Gamma}}{\sketch{p}}{\sketch{t}}$) and thus extracts the relation between variables as a substitution $\sigma$.
    \item Parts 2 and 3 are introduced regarding the two premises in the \textsc{T-App} rule.
    Part 2 completes the first premise $\stlctype{\Gamma}{p_1}{t_1 \rightarrow t_2}$, resulting in an assignment $\sigma_1$.
    Then, Part 3 combines the known information by composing $\sigma$ and $\sigma_1$ into $\sigma'$, and further completes the second premise $\stlctype{\Gamma}{p_2}{t_1}$.
    In general, this process repeats for all premises until each is completed, and any failure to discharge a premise causes the entire rule application to fail.
\end{itemize}

To get the result, \mainname composes the two obtained substitutions $\sigma_2$ and $\sigma'$.
Note that this composition involves not only variables in $\stlctype{\sketch{\Gamma}}{\sketch{p}}{\sketch{t}}$ but also local variables such as $\sk{\Gamma}$ and $\sk{p_1}$. 
Therefore, \textsc{S-App} will exclude these local variables and return only values of interest.

This translation establishes a direct correspondence between typing rules and synthesis rules. The core insight is that when the program is unknown, we can apply synthesis rules to construct an existential proof, which simultaneously discovers the program and proves its type correctness. The resulting existential proof has the same structure as the type correctness proof of the discovered program (like the one depicted in \autoref{figure:sample-type-tree}).

\subsection{Interacting with the Language Model}

Having defined the synthesis rules, we now describe how the LM interacts with the synthesis system.

\noindent \textbf{LM Queries.} The synthesis rules form a search system for program synthesis, where the system issues queries to the LM and recursively discharges synthesis subgoals based on the decisions the LM makes. 
As discussed in \autoref{subsection:moti-representation}, at each synthesis step, the system raises two types of queries: to select an appropriate synthesis rule, and to generate variable assignments for predicates. In \mainname, we resolve both queries by prompting the LM with the following information:
\begin{enumerate}
    \item the natural language prompt of the overall synthesis task
    \item the sequence of synthesis decisions generated so far
    \item the current synthesis goal
\end{enumerate}

To simplify the interaction, we do not distinguish the two types of queries and assume that the LM can learn to infer what information is needed in each step from the context.
\autoref{fig:synthesis-sequence} illustrates the expected decision sequence for our sample task.
Let us consider the first three steps:
\begin{itemize}
\item In Step 1, the LM receives the task prompt, an empty decision sequence, and the overall synthesis goal $\stlctype{\tcode{empty}}{\sk{p}}{\sk{t}}$. It predicts the synthesis rule \textsc{S-App} with no other output, as no variable assignment is needed in this rule application.
\item In Step 2, the case is similar. The LM predicts only \textsc{S-Abs} with no other assignment.
\item In Step 3, the decision sequence becomes $[\textsc{S-App}, \textsc{S-Abs}]$ and the subgoal becomes $\stlctype{\snoc{\tcode{empty}}{\sk{x_1}}{\sk{t_3}}}{\sk{p_3}}{\sk{t_4}}$. Here, the LM predicts the \textsc{S-Var} rule, and meanwhile, generates all four assignments (\autoref{formula:sig4}) required by this rule.
\end{itemize}

\begin{figure}[tb]
\centering
\begin{tikzpicture}[x=1.2cm, y=0.8cm]
    \tikzset{
        rule/.style={
            draw=gray!80, fill=gray!15,
            minimum height=5mm, minimum width=10mm,
            rounded corners=1.5pt, font=\small,
            text depth=0.25ex, text height=1.6ex
        },
        assign/.style={
            draw=blue!40, fill=blue!8,
            minimum height=5mm, minimum width=6mm,
            rounded corners=1.5pt, font=\small,
            text depth=0.25ex, text height=1.6ex
        },
    }

    \node[rule] at (0,0) {S-App};
    \node[rule] at (1,0) {S-Abs};
    \node[rule] (a0) at (2,0) {S-Var};

    \node[assign] (a1) at (2.80,0) {\tcode{x}};    
    \node[assign] (a2) at (3.45,0) {\tcode{x}};    
    \node[assign] (a3) at (4.20,0) {\tcode{bool}}; 
    \node[assign] (a4) at (5.05,0) {\tcode{bool}}; 

    \node[rule] at (6.0,0) {$\boldsymbol{\cdots}$};

    \node[above] at (0,0.3) {\tiny Step 1};
    \node[above] at (1,0.3) {\tiny Step 2};
    \node[above] at (3.4,0.3) {\tiny Step 3};
    \node[above] at (6.0,0.3) {\tiny Step 4};

    \node[
        draw=gray!60,
        dashed,
        rounded corners=2pt,
        fit={(a0) (a4)},
        inner sep=2pt
    ] {};
\end{tikzpicture}

\setlength{\fboxsep}{1pt}
\caption{Synthesis decision sequence for $\app{(\abs{\tcode{x}}{\tcode{bool}}{\tcode{x}})}{\tcode{true}}$: 
\colorbox{gray!15}{\textsf{synthesis rules}} and 
\colorbox{blue!8}{\textsf{variable assignments}}.}
\label{fig:synthesis-sequence}
\end{figure}

As shown in \autoref{fig:synthesis-sequence}, the synthesis rules and variable assignments together form a synthesis decision sequence, which serves as a novel program representation in \mainname. 
Unlike traditional representations that treat programs as flat sequences of syntactic tokens (e.g., [$\lambda$, $x$, :, \tcode{bool}, $\dots$]), our representation is structured around the type correctness proof. Each token corresponds to a step in the proof construction process, directly capturing the hierarchical structure of type reasoning (as in \autoref{figure:sample-type-tree}). This enables the LM to learn the type system explicitly rather than attempting to reconstruct it from flat token sequences. Since this sequence represents a program synthesis process that determines both the resulting program and its type correctness proof, this representation embodies our \textbf{Derivation Vicinality} principle.

To support this proof-guided synthesis, we design a novel model architecture with dual-encoding: the encoder processes both the static natural language specification and the dynamically evolving synthesis goal at each step, providing rich typing context to the decoder. Details of this architecture are presented in \autoref{sec:implementation}.

\smallskip \noindent \textbf{Training.} Since our representation differs from traditional token sequences, the LM requires fine-tuning to generate synthesis decision sequences.
Our representation also offers \textbf{Data Usability}, as it is compatible with existing code generation datasets for the target language.
Given a corpus of programs, we construct their type correctness proofs via the type checker, convert the proofs into synthesis decision sequences using the correspondence between typing rules and synthesis rules, and obtain a training set in \mainname's representation.

\section{Translation from Typing Rules to Synthesis Rules \label{sec:meta}}

In this section, we introduce the necessary preliminaries on constrained Horn clauses and unification, followed by the user-provided specifications of the object language, including terms, typing rules, and describe how typing rules are systematically translated into synthesis rules.

\subsection{Preliminaries}
\subsubsection{Constrained Horn Clauses}
\label{chc-intro}
Constrained Horn clauses (CHCs) are a fragment of first-order logic with applications to program verification and synthesis~\cite{DBLP:conf/synasc/GurfinkelB19}. A constrained Horn clause, in this setting, is a universally quantified implication of the canonical form:
\begin{center}
\begin{prooftree}[rule margin=0.5ex, separation=1em]
\hypo{P_1(\overline{x}_1)}
\hypo{P_2(\overline{x}_2)}
\hypo{\ldots}
\hypo{P_n(\overline{x}_n)}
\hypo{\phi(\overline{y})}
\infer5[(CHC)]{P(\overline{x}_0)}
\end{prooftree}
\captionof{figure}{General form of a constrained Horn clause.}
\label{T-Rule}
\end{center}
Each clause consists of zero or more premises, where each premise is a judgment formed by applying an uninterpreted predicate $P_i$ to terms $\overline{x}_i$. It also includes a constraint $\phi(\overline{y})$, a computable function over terms $\overline{y}$, and concludes with a single judgment $P(\overline{x}_0)$.
This rule states that the judgment $P(\overline{x}_0)$ can be derived whenever all the premises $P_i(\overline{x}_i)$ are derivable and the constraint $\phi(\overline{y})$ holds. Without loss of generality, we can assume that every variable of terms in $\overline{y}$ occurs among a term in the premises or conclusion, namely $FV(\overline{y})\subseteq \bigcup_{i=0}^n FV(\overline{x}_i)$, which ensures that the constraints are always deterministic and checkable given the premises and conclusion.

\begin{example}
Consider the typing rule \textsc{T-App} for function application in $\stlc$:
\begin{center}
\begin{prooftree}[rule margin=1ex, separation=1.5em]
\hypo{\stlctype{\Gamma}{p_1}{t_1 \rightarrow t_2}}
\hypo{\stlctype{\Gamma}{p_2}{t_1}}
\infer2[(T-App)]{\stlctype{\Gamma}{p_1~p_2}{t_2}}
\end{prooftree}
\end{center}
In CHC terminology: the predicates $P_1$ and $P_2$ are both the typing judgment predicate $\stlctype{\cdot}{\cdot}{\cdot}$; the terms $\overline{x}_1 = (\Gamma, p_1, t_1 \rightarrow t_2)$ and $\overline{x}_2 = (\Gamma, p_2, t_1)$ are the arguments to the premises; the conclusion $P(\overline{x}_0)$ is $\stlctype{\Gamma}{p_1~p_2}{t_2}$ with $\overline{x}_0 = (\Gamma, p_1~p_2, t_2)$; and the constraint $\phi$ is implicitly $\mathit{True}$.
\end{example}

It is well known that any Turing-computable function can be encoded using Horn clauses (HCs)~\cite{DBLP:journals/bit/Tarnlund77}. CHCs extend HCs by adding constraints $\phi(\overline{y})$, which allow expressing computable conditions (e.g., arithmetic comparisons, equality checks) directly within the clause. When all constraints are interpreted as $\mathit{True}$, CHCs collapse to ordinary HCs. Consequently, CHCs are at least as expressive as HCs, and any Turing-computable function can be expressed using CHCs, which ensures the expressiveness of our framework.

\subsubsection{Unification}
Unification is a fundamental operation in logic and computer science that finds substitutions to make two expressions structurally identical. Formally, given two terms $s$ and $t$, a \emph{unifier} is a substitution $\sigma$ such that $\sigma(s) = \sigma(t)$. When a unifier exists, the two terms are said to be \emph{unifiable}. Among all possible unifiers, the \emph{most general unifier} (MGU) provides the least restrictive solution: any other unifier can be obtained by further instantiating the MGU.

\begin{example}
Consider unifying an incomplete typing judgment $\stlctype{\tcode{empty}}{\sk{p_1}}{\arr{\sk{t_1}}{\sk{t_2}}}$ with another judgment $\stlctype{\sk{\Gamma_2}}{\left(\abs{\sk{x_1}}{\sk{t_3}}{\sk{p_3}}\right)}{\arr{\sk{t_3}}{\sk{t_4}}}$. The unification process matches corresponding positions: $\tcode{empty}$ with $\sk{\Gamma_2}$, $\sk{p_1}$ with $\abs{\sk{x_1}}{\sk{t_3}}{\sk{p_3}}$, and $\arr{\sk{t_1}}{\sk{t_2}}$ with $\arr{\sk{t_3}}{\sk{t_4}}$. This yields the MGU:
\[
\sigma = \{\sk{\Gamma_2} \mapsto \tcode{empty},\quad\sk{p_1} \mapsto \abs{\sk{x_1}}{\sk{t_3}}{\sk{p_3}},\quad\sk{t_1} \mapsto \sk{t_3},\quad\sk{t_2} \mapsto \sk{t_4}\}
\]
After applying $\sigma$, both terms become identical: $\stlctype{\tcode{empty}}{\left(\abs{\sk{x_1}}{\sk{t_3}}{\sk{p_3}}\right)}{\arr{\sk{t_3}}{\sk{t_4}}}$.

Note that unification is symmetric: it can bind variables from either term. Here, $\sk{\Gamma_2}$ is bound to a concrete value $\tcode{empty}$, while $\sk{t_1}$ and $\sk{t_2}$ are bound to other variables $\sk{t_3}$ and $\sk{t_4}$.
The MGU preserves maximum generality by not over-constraining variables. For instance, the substitution $\{\sk{\Gamma_2} \mapsto \tcode{empty}, \sk{p_1} \mapsto \abs{\sk{x_1}}{\sk{t_3}}{\sk{p_3}}, \sk{t_1} \mapsto \sk{t_3}, \sk{t_2} \mapsto \sk{t_4}, \sk{t_3} \mapsto \tcode{bool}\}$ is also a valid unifier, but it is not the MGU because it unnecessarily constrains $\sk{t_3}$ to $\tcode{bool}$.
\end{example}

In our synthesis framework, unification enables typing rules to be applied by aligning the schematic parameters in a rule's conclusion with the current synthesis goal. When a typing rule is selected, unification extracts the relationship between the goal's unknowns and the rule's parameters, propagating known information and establishing constraints between unknowns.
In this paper, we restrict ourselves to \emph{first-order unification}, where variables range only over terms and not over functions or predicates. 
First-order unification is decidable and admits efficient algorithmic solutions~\cite{DBLP:journals/jacm/Robinson65}, making it suitable for our synthesis framework.

\subsection{Terms and Typing Rules}
\subsubsection{Terms}
\label{terms-intro}
We begin by defining the object language through a structured set of terms. A \emph{term} is a symbolic expression that represents programs, types, contexts, or other syntactic entities in the object language. Each term type is specified inductively: terms are constructed from constants, variables, and composite structures via constructors. For any given term type \tcode{T}, the syntax conforms to the general form:
\[\tcode{T} \,::=\, \tcode{c} \;\mid\; \tcode{v} \;\mid\; \tcode{k}(\tcode{T}_1, \ldots, \tcode{T}_n)\]
where \(\tcode{c}\) ranges over constants \(\mathcal{C}_\scode{T}\), \(\tcode{v}\) ranges over the global set of variables \(\mathcal{V}\), and \(\tcode{k}\) ranges over constructor symbols \(\mathcal{K}_\scode{T}\) of specified arity and parameter term types. 

The variable set $\mathcal{V}$ consists of symbolic names serving as placeholders to be instantiated by terms. In contrast, the sets of constants $\mathcal{C}_\scode{T}$ and constructors $\mathcal{K}_\scode{T}$ are determined by user-defined inductive type declarations. 
Built-in types such as \tcode{Int} and \tcode{String} are provided following conventional programming language syntax.

\begin{example}
Consider the following syntax for programs in $\stlc$:
\[
\tcode{Prog} ::= 
  \tcode{true} \;\mid\;
  \tcode{false} \;\mid\;
  \tcode{var}(\tcode{String}) \;\mid\;
  \tcode{app}(\tcode{Prog}, \tcode{Prog}) \;\mid\;
  \tcode{abs}(\tcode{String}, \tcode{Type}, \tcode{Prog})
\]
This yields the following instantiations of the general term pattern:
\begin{itemize}
  \item \textbf{Constants:} \(\mathcal{C}_{\tcode{Prog}} = \{\tcode{true}, \tcode{false}\}\)
  \item \textbf{Constructors:}
    \(\mathcal{K}_{\text{Prog}} = \{
      \tcode{var} : (\tcode{String}),\ 
      \tcode{app} : (\tcode{Prog}, \tcode{Prog}),\ 
      \tcode{abs} : (\tcode{String}, \tcode{Type}, \tcode{Prog})
    \}\)
\end{itemize}
Using this syntax, the program $\lambda x{:}\tcode{bool}.\, x$ is represented as the term $\tcode{abs}(\tcode{"x"}, \tcode{bool}, \tcode{var}(\tcode{"x"}))$. This term is constructed by applying the constructor $\tcode{abs}$ to three arguments: the string $\tcode{"x"}$, the type $\tcode{bool}$, and the subterm $\tcode{var}(\tcode{"x"})$.
\end{example}

\subsubsection{Categories of terms and substitutions} 
In general, terms may contain known and unknown components, so in the remainder of this paper, we do not explicitly annotate terms as in \autoref{sketch-label}. Below, we provide definitions for key concepts that will be used throughout the subsequent sections.
\begin{enumerate}
    \item \textbf{Ground terms} are terms constructed exclusively from constants and constructors, containing no variables. 
    They are fully specified expressions that require no further instantiation.
    
    \item A \textbf{substitution} is a mapping \(\sigma: \mathcal{V}_0 \to \mathcal{T}\), where \(\mathcal{V}_0\) is a subset of variables \(\mathcal{V}\), and \(\mathcal{T}\) is the set of all terms. A substitution replaces each variable \(v \in \mathcal{V}_0\) with a term \(\sigma(v) \in \mathcal{T}\).

    The \textbf{composition} of two substitutions $\sigma_i$ and $\sigma_j$, denoted $\sigma_i \sigma_j$, is defined such that for any variable $v \in \mathcal{V}$, $(\sigma_i \sigma_j)(v) = \sigma_i(\sigma_j(v))$. 
    This allows multiple substitutions to be applied sequentially.
    A substitution $\sigma$ can also operate on a list of terms $\overline{t} = [t_1, t_2, \ldots, t_n]$, yielding a new list $\sigma(\overline{t}) = [\sigma(t_1), \sigma(t_2), \ldots, \sigma(t_n)]$ where $\sigma$ is applied to each term individually.
    
    \item An \textbf{assignment} is a special case of a substitution whose codomain ranges over ground terms.
    Assignments serve as the ultimate objective of the synthesis, as we require that the typing judgments, as synthesis goals, can be instantiated with no remaining uncertainties.
\end{enumerate}

\subsubsection{Typing Rules and Type Derivation Trees}
We formalize the typing rules of the object language as CHCs for their expressiveness and standard structure, where each typing rule is provided by the user as a CHC shown in \autoref{T-Rule}.
The user should also provide at least one typing rule whose conclusion is $\welltyped{p}$, the top-level judgment asserting that program $p$ is well-typed. For example, in $\stlc$, this judgment can be defined as $\exists t.\; \stlctype{\tcode{empty}}{p}{t}$, meaning that the program $p$ is well-typed under the empty context.

A judgment $P(\overline{t})$, where $\overline{t}$ are ground terms, is \emph{derivable} from a given set of typing rules if there exists a finite derivation tree whose root is labeled with $P(\overline{t})$, each internal node corresponds to the application of a typing rule, and all constraints of typing rules are satisfied. The type derivation tree is precisely this tree structure that witnesses the derivability of a typing judgment.
From a logical perspective, a type derivation tree is a \textbf{type correctness proof}: it proves that the typing judgment at the root holds. This correspondence between derivation trees and proofs is central to our approach. Later, we will show that a synthesis derivation tree corresponds to constructing an \textbf{existential type correctness proof} when the program is unknown: the synthesis process simultaneously discovers the program and constructs a proof of its type correctness.

More formally, a \emph{type derivation tree} for a typing judgment $P(\overline{t})$ is defined as: (1) the root node is labeled with $P(\overline{t})$; (2) each node is labeled with a typing judgment $P_i(\overline{t_i})$, where all terms $\overline{t_i}$ are ground terms, and annotated with the typing rule which derives $P_i(\overline{t_i})$ under certain instantiation; (3) the children of each internal node correspond to the premises of the applied typing rule; and (4) for each applied typing rule, all associated constraints are satisfied under the instantiation.

For example, \autoref{figure:sample-type-tree} shows a type derivation tree for the program $(\lambda x{:}\tcode{bool}.\, x)\; \tcode{true}$, which proves that this program has type $\tcode{bool}$ under the empty context. Consider the following subtree for the subterm $\lambda x{:}\tcode{bool}.\, x$:
\begin{center}
\begin{prooftree}[rule margin=1ex, separation=1em]
\hypo{\checkBinding{x{:}\tcode{bool}}{x}{\tcode{bool}}}
\infer1[(T-Var)]{\stlctype{x{:}\tcode{bool}}{x}{\tcode{bool}}}
\infer1[(T-Abs)]{\stlctype{\tcode{empty}}{\lambda x{:}\tcode{bool}.\, x}{\tcode{bool} \rightarrow \tcode{bool}}}
\end{prooftree}
\end{center}
In this subtree, the root node is labeled with the typing judgment $\stlctype{\tcode{empty}}{\lambda x{:}\tcode{bool}.\, x}{\tcode{bool} \rightarrow \tcode{bool}}$ and annotated with the typing rule \textsc{T-Abs}. Its child node is labeled with $\stlctype{x{:}\tcode{bool}}{x}{\tcode{bool}}$ and annotated with \textsc{T-Var}, which has a constraint $\checkBinding{x{:}\tcode{bool}}{x}{\tcode{bool}}$ that checks whether the variable $x$ with type $\tcode{bool}$ exists in the context.

The existence of a type derivation tree for $\welltyped{p}$ establishes that the program $p$ is well-typed. To enable this verification, there should be a type checker in the object language toolchain that constructs the type derivation tree for $\welltyped{p}$ given any well-typed program $p$.

\subsection{Synthesis rules}
\label{sec:meta-translation}

\subsubsection{Synthesis Judgments and Synthesis Goals}
As outlined in the overview, when some components of a typing judgment are unknown (e.g., the program to be synthesized), typing rules can be reformulated as synthesis rules to handle such uncertainties. 
A typing judgment $R(\overline{t})$ is converted into a \emph{synthesis judgment} $R(\overline{t}) \leadsto \theta$, where $R(\overline{t})$ denotes the \emph{synthesis goal}, and $\theta$ is the assignment synthesized for the goal. 
Semantically, $R(\overline{t}) \leadsto \theta$ is derivable iff, after applying $\theta$ to $\overline{t}$, all terms in $\theta(\overline{t})$ are ground and the typing judgment $R(\theta(\overline{t}))$ is derivable.

\subsubsection{Translation from Typing Rules into Synthesis Rules}
\begin{center}
\begin{prooftree}[rule margin=1ex, separation=1.5em]
\hypo{\unify{\overline{x}=\overline{x}_0}=\sigma}
\infer[no rule]1{\acquire{\sigma(\overline{x}_f)} = \sigma_0}
\hypo{P_1(\sigma_0\sigma(\overline{x}_1)) \leadsto \sigma_1}
\infer[no rule]1{P_2(\sigma_1\sigma_0\sigma(\overline{x}_2)) \leadsto \sigma_2}
\ellipsis{}{P_n(\sigma_{n-1}\ldots\sigma_0\sigma(\overline{x}_n))\leadsto \sigma_n}
\hypo{\ensure{\phi(\sigma_n\ldots\sigma_0\sigma(\overline{y}))}}
\infer3[(S-Rule)]{P(\overline{x}) \leadsto \sigma_n\ldots\sigma_0\sigma|_{\mvhzc(\overline{x})}}
\end{prooftree}
\captionof{figure}{Normal form of the synthesis rules}
\label{S-Rule}
\end{center}

A typing rule in \autoref{T-Rule} can be translated into a synthesis rule in \autoref{S-Rule} in the following manner:
\begin{enumerate}
  \item \textbf{Conclusion translation.} The conclusion $P(\overline{x}_0)$ of the typing rule is replaced by a synthesis judgment $P(\overline{x}) \leadsto \sigma_n \ldots \sigma_0 \sigma|_{\mvhzc(\overline{x})}$ (with the same predicate $P$), where $\overline{x}$ is a sequence of newly introduced variables and can be instantiated with respect to the synthesis goal. $\mvhzc(\overline{x})$ denotes the set of variables in $\overline{x}$, which are the only variables that need to be synthesized. 
  $\sigma_n \ldots \sigma_0 \sigma|_{\mvhzc(\overline{x})}$ indicates that the final assignment is restricted to only those variables in $\overline{x}$.

  \item \textbf{Adding unification.} The unification step is added at the beginning of the synthesis rule, which unifies the variables $\overline{x}$ in the synthesis goal with the variables $\overline{x}_0$ in the conclusion of the typing judgment, yielding a substitution $\sigma$ that is the MGU of the unification problem. 
  
  \item \textbf{Adding Free variables acquisition.} The variables appearing in $\overline{x}_0$ but not in any of the $\overline{x}_i$ are collected into $\overline{x}_f$, defined as $\overline{x}_f = \text{list}(\mvhzc(\overline{x}_0) \setminus \bigcup_{i=1}^n \mvhzc(\overline{x}_i))$. We refer to $\overline{x}_f$ as the \emph{free variables} of the synthesis rule: they may not be bound by any premise even after unification, and thus should be acquired from external sources, yielding the assignment $\sigma_0$.
  
  \item \textbf{Premises translation.} Each premise $P_i(\overline{x}_i)$ is replaced by $P_i(\sigma_{i-1} \cdots \sigma_0 \sigma(\overline{x}_i)) \leadsto \sigma_i$, where $\sigma_{i-1} \cdots \sigma_0 \sigma(\overline{x}_i)$ instantiates the original terms $\overline{x}_i$ with substitutions from previous steps, and $\sigma_i$ is the assignment synthesized for the subgoal $P_i(\sigma_{i-1} \cdots \sigma_0 \sigma(\overline{x}_i))$ ($i\geq1$).
  
  \item \textbf{Constraint translation.} The constraint $\phi(\overline{y})$ is translated into $\ensure{\phi(\sigma_n \ldots \sigma_0 \sigma(\overline{y}))}$, which checks whether it holds under the assignment $\sigma_n \ldots \sigma_0 \sigma$.    
\end{enumerate}

\begin{example}
For example, if we have the following typing rule:
\begin{center}
\begin{prooftree}[rule margin=1ex, separation=1em]
\hypo{P_1(x_1)}
\hypo{P_2(x_2, c(x_1,x_3))}
\hypo{\phi(x_0, x_3)}
\infer3[(T-Rule1)]{P(x_0, c(x_1,x_2))}       
\end{prooftree}
\end{center}
We can translate it into the following synthesis rule:
\begin{center}
\begin{prooftree}[rule margin=1ex, separation=0.7em]
\hypo{\unify{s_0=x_0,\, s_1=c(x_1, x_2)}=\sigma}
\infer[no rule]1{\acquire{\sigma(x_0)} = \sigma_0}
\hypo{P_1(\sigma_0\sigma(x_1)) \leadsto \sigma_1}
\infer[no rule]1{P_2(\sigma_1\sigma_0\sigma(x_2,\, c(x_1,x_3))) \leadsto \sigma_2}
\hypo{\ensure{\phi(\sigma_2\sigma_1\sigma_0\sigma(x_0, x_3))}}
\infer3[(S-Rule1)]{P(s_0, s_1) \leadsto \sigma_2\sigma_1\sigma_0\sigma|_{\mvhzc(s_0, s_1)}}
\end{prooftree}
\end{center}
\end{example}

\begin{lemma}
\label{rule-bijection}
There is a one-to-one correspondence between typing rules and synthesis rules.
\end{lemma}

\begin{proof}
The complete proof can be found in Appendix~\ref{sec:appendix-lemma1}.
\end{proof}

\section{Synthesis System}
\label{sec:approach}
\autoref{fig:classfy-problem} presents the architecture of \mainname's system, which constructs synthesis derivation trees from which programs are extracted. The construction is achieved in \autoref{alg:synthtree} through a series of queries (\autoref{sec:synthtree-construction-query}), which can be resolved through two approaches corresponding to the dual requirements of our system: code generation (\autoref{sec:query-resolution-lm}) and training data preparation (\autoref{sec:translation}).

The central claim of this section is that constructing a synthesis derivation tree, which is the core of \mainname's framework, is equivalent to constructing an existential type correctness proof. A key insight underlying this equivalence is the isomorphism between type derivation trees and synthesis derivation trees (\autoref{sec:isomorphism}). 
This isomorphism ensures that the constructed synthesis derivation trees maintain the logical flow and dependencies inherent in the type-level reasoning, so that we can construct synthesis derivation trees to represent well-typed programs (\autoref{sec:soundness-completeness}).

\begin{figure}[t]
\centering
\includegraphics[width=0.75\textwidth]{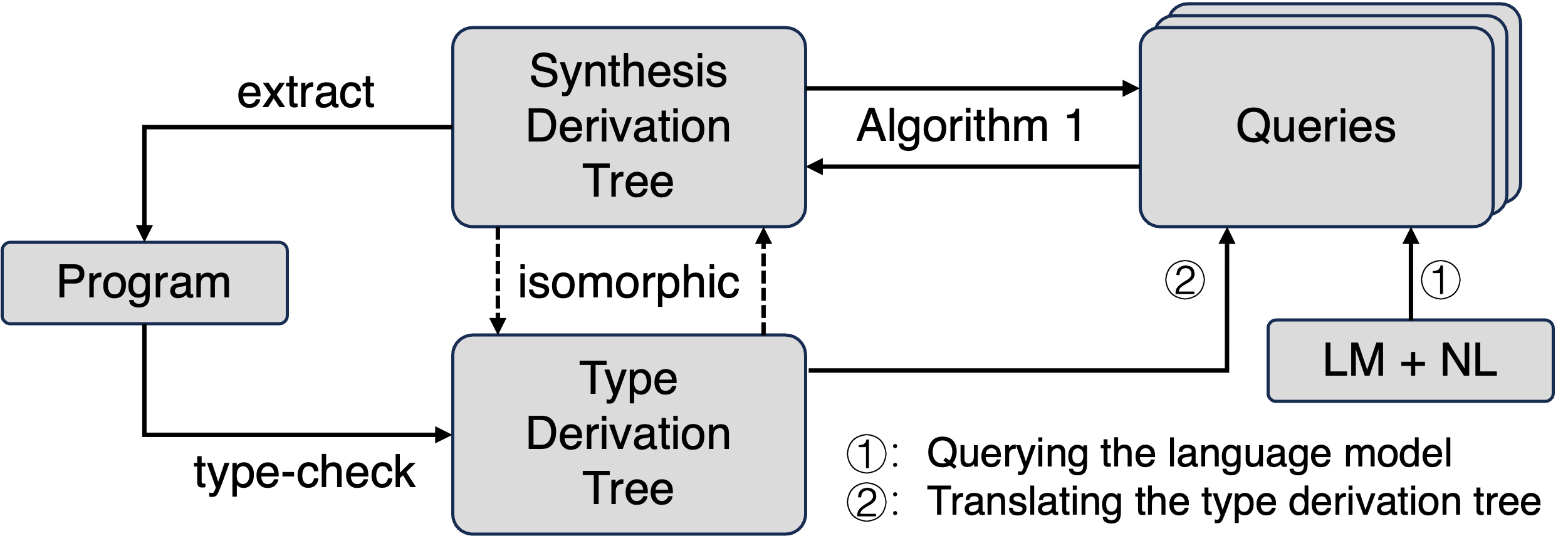}
\caption{Meta-system architecture of \mainname.}
\label{fig:classfy-problem}
\end{figure}

\subsection{Construction of Synthesis Derivation Trees}
In this subsection, we will first describe the application of synthesis rules (\autoref{sec:synthesis-rule-application}), then introduce the concept of synthesis derivation trees (\autoref{sec:synthesis-derivation-tree}), and finally present the algorithmic construction of synthesis derivation trees, where it is realized by a sequence of query resolutions (\autoref{sec:synthtree-construction-query}).
\subsubsection{Synthesis Rule Application}
\label{sec:synthesis-rule-application}
Each synthesis rule as in \autoref{S-Rule} specifies how to start from a partially instantiated typing judgment $P(\overline{x})$, search for an assignment $\sigma$ to all variables in $\overline{x}$, then derive $P(\sigma(\overline{x}))$ using the corresponding typing rule. This procedure can be detailed as follows:
\begin{enumerate}
    \item \textbf{Unification Step:} First unify the terms in the synthesis goal $P(\overline{x})$ with the terms in the conclusion of the typing rule, yielding the MGU $\sigma$. Once the unification is done, the other terms ($\overline{x}_i$, $\overline{x}_f$, and $\overline{y}$) should also be instantiated with $\sigma$ to match the synthesis goal. Or the unification fails, in which case the synthesis goal is not derivable using this synthesis rule.
    
    
    \item \textbf{Acquisition Step:} Next acquire the assignment for all free variables $\overline{x}_f$ in the synthesis rules (after instantiated by $\sigma$) from external sources. 
    \mainname can infer the term type of each variable to be synthesized based on the definitions of predicates and terms. The generated assignments must conform to the inductively defined syntax described in \autoref{terms-intro}; if any constant or constructor does not match the required term type, the acquisition step fails.
    
    
    \item \textbf{Sub-Derivation Steps:} Then, synthesize assignments for all subgoals in sequence. For each subgoal $P_i(\sigma_{i-1}\ldots\sigma_0\sigma(\overline{x}_i)) \leadsto \sigma_i$, begin with an instantiation for variables in $\overline{x}_i$ with the substitutions $\sigma_{i-1}, \ldots, \sigma_0, \sigma$ in previous steps. 
    This refines the subgoal and propagates synthesis information from earlier subgoals. Subsequently, synthesize the assignment $\sigma_i$ for the current subgoal by recursively applying synthesis rules.
    
    \item \textbf{Constraint Checking Step:} After all synthesis subgoals have been solved, \mainname checks whether the constraint $\phi(\overline{y})$  is satisfied under the combined assignment $\sigma_n \ldots \sigma_0 \sigma$. 
    
    \item \textbf{Assignment Collection Step:} In the final step, \mainname compose all substitutions obtained during synthesis to form $\sigma_c = \sigma_n \ldots \sigma_0 \sigma$, and restrict this assignment to only the variables occurring in $\overline{x}$, yielding the overall assignment $\theta = \sigma_c|_{\mvhzc(\overline{x})}$. 
\end{enumerate}

In steps (3) and (5), we assert without proof that the composed substitution $\sigma_c = \sigma_n \ldots \sigma_0 \sigma$ and the final collected $\theta$ are assignments. Specifically, $\sigma_c$ maps all variables in $\overline{x}$ and $\overline{y}$ to ground terms. This property can be formally established by induction on the structure of the synthesis derivation.

\begin{lemma}
\label{lemma:assignment}
For any successful application of the synthesis rule in \autoref{S-Rule}, the substitution $\sigma_c = \sigma_n \ldots \sigma_0 \sigma$ is an assignment, in the sense that $\forall i$, $\sigma_c(\skhzc{\overline{x}_i})$, $\sigma_c(\skhzc{\overline{x}})$, and $\sigma_c(\skhzc{\overline{y}})$ are all ground terms.
\end{lemma}
\begin{proof}
  The complete proof can be found in Appendix~\ref{sec:appendix-lemma2}.
\end{proof}

\subsubsection{Synthesis Derivation Trees}
\label{sec:synthesis-derivation-tree}
To formalize the structure of synthesis derivations, we introduce the concept of synthesis derivation trees in analogy to type derivation trees in type systems. 

A \emph{synthesis derivation tree} is a rooted tree that witnesses the derivability of a synthesis judgment through the application of synthesis rules. More formally, a synthesis derivation tree for a synthesis judgment $P(\skhzc{\overline{t}}) \leadsto \theta$ is a tree where: (1) the root node is labeled with $P(\skhzc{\overline{t}}) \leadsto \theta$; (2) each node is labeled with a synthesis judgment and annotated with both the synthesis rule used and the assignment for its free variables; (3) the children of each internal node correspond to the subgoals of the annotated synthesis rule; (4) for each synthesis rule, all steps including unification, assignment acquisition, subgoals synthesis and constraint checking succeed.

By Lemma~\ref{lemma:assignment}, the existence of a synthesis derivation tree for a synthesis judgment $P(\overline{t}) \leadsto \theta$ guarantees that there is an assignment $\theta$ satisfying the synthesis goal $P(\overline{t})$, that is, $\theta$ maps all variables to ground terms and $P(\sigma(\overline{t}))$ is derivable with respect to typing rules.

\subsubsection{Algorithmic Construction of Synthesis Derivation Trees}
\label{sec:synthtree-construction-query}
The algorithmic construction of synthesis derivation trees is the main focus of \mainname, whether through the interaction with LMs, or via the translation from type derivation trees.
For a synthesis goal of the form $P(\overline{t})$, the task is to synthesize an assignment $\theta$ and construct a synthesis derivation tree with the root node labeled $P(\overline{t}) \leadsto \theta$ by recursively applying synthesis rules as described in \autoref{sec:synthesis-rule-application}.


\begin{algorithm}[tb]
  \caption{Synthesis Derivation Tree Construction}
  \label{alg:synthtree}
  \KwIn{A synthesis goal $P(\skhzc{\overline{t}})$}
  \KwOut{A synthesis derivation tree $T$ whose root node is labeled $P(\skhzc{\overline{t}}) \leadsto \theta$}
  \SetKwFunction{GenSynthTree}{GenSynthTree}
  \SetKwProg{Fn}{Function}{:}{}
  
  \Fn{\GenSynthTree{$P(\overline{t})$}}{
    Select a synthesis rule $R$ with conclusion $P(\overline{x}) \leadsto \ldots$
    \tcp*{Suppose $R$ is as \autoref{S-Rule}}
    
    Instantiate $\overline{x}$ with $\overline{t}$ in $R$\;
    $\sigma \gets$ MGU from the unification step of $R$\;
    \lIf{$\sigma = \bot$}{\Return{$\bot$}} 

    $\sigma_0 \gets \func{AcquireFreeVariableAssignment}(\sigma(\overline{x}_f))$\;
    \lIf{$\sigma_0 = \bot$}{\Return{$\bot$}}
    
    \ForEach{\textnormal{subgoal} $P_i(\sigma_{i-1} \ldots \sigma_1 \sigma(\overline{x}_i))$ \textnormal{of} $R$}
    {
        $T_i \gets \GenSynthTree{ $P_i(\sigma_{i-1} \ldots \sigma_1 \sigma(\overline{x}_i))$} $\;
        \lIf{$T_i = \bot$}{\Return{$\bot$}}
        $\sigma_i \gets$ assignment from the root node of $T_i$\;
    }
    
    \lIf{$\phi(\sigma_n \ldots \sigma_0 \sigma(\overline{y})) = \bot$}{\Return{$\bot$}}
    $\theta \gets \sigma_n \ldots \sigma_0 \sigma|_{\mvhzc(\overline{t})}$\;
    \Return{$\func{MakeTreeNode}(P(\overline{t}) \leadsto \theta, R, \sigma_0, T_1, \ldots, T_n)$}
}
\end{algorithm}

This process is detailed in \autoref{alg:synthtree}. In this algorithm, only two functions act as black boxes: the selection of synthesis rules and the acquisition of assignments for free variables. Both can be viewed as query problems. Synthesis rule selection is a single-step query that chooses one rule from a finite set of candidates, while assignment acquisition is a multi-step query that constructs a ground term for each free variable using a finite set of constructors and constants.

Consequently, our approach reduces the synthesis problem to finding appropriate solvers for these two specific query problems. This makes our framework more general and modular, as the extracted query resolution problems can be implemented independently. 
Moreover, the in-time checking mechanisms (unification and constraint verification, etc.) in \autoref{alg:synthtree} significantly restrict the search space compared to naive enumeration-based approaches. 

\subsection{Query Resolution}
\autoref{fig:classfy-problem} illustrates how query resolution serves as a central mechanism in constructing a synthesis derivation tree. Queries should be resolved to determine which synthesis rule to apply and how to instantiate variables. In the following, we present two complementary approaches for query resolution: employing LMs to make decisions during code generation (\autoref{sec:query-resolution-lm}), and translating type derivation trees to prepare supervised training data (\autoref{sec:translation}).

\subsubsection{Query Resolution using LMs}
\label{sec:query-resolution-lm}
In the code generation scenario, both query resolution functions are delegated to the LM, which, given the natural language input, acts as a decision-making component for synthesis rule selection and variable assignment.

\begin{itemize}
  \item \textbf{Selection of synthesis rules:} Given the natural language description, the current partially constructed synthesis derivation tree, and the synthesis goal, the LM is expected to choose a specific synthesis rule from the available rule set.
  
  \item \textbf{Assignment of free variables:} After a synthesis rule is selected and unification is performed, the LM is required to generate a sequence of constructors and constants, constructing ground terms for assigning free variables. 
\end{itemize}

LM prediction is performed in a contextually aware manner, conditioned on multiple sources of information: (1) the natural language describing the desired program behavior; (2) the current synthesis derivation tree, which encodes both the ongoing synthesis process and the type derivation process via the isomorphism between type and synthesis derivation trees; and (3) the synthesis goal $P(\overline{t})$, which includes already synthesized program fragments and the current typing context.

The LM leverages its learned type-reasoning ability to generate the most appropriate synthesis rule and corresponding variable assignments by jointly reasoning about the semantic intentions conveyed by the natural language and the structural requirements encoded in the derivation tree. 

\subsubsection{Query Resolution through Type Derivation Tree Translation}
\label{sec:translation}
Beyond the code generation scenario in which natural language serves as input, \mainname requires curated training data to teach the LM how to interpret and solve synthesis queries. 
We generate such queries by translating existing type derivation trees of target programs, which involves recording the contextual information and the ground truth decision at each step of synthesis rule application.

The translation is performed by traversing the type derivation tree. When we visit its root node, we begin constructing the synthesis derivation tree from the root node too. 
For the $i$-th child labeled $P_i(\overline{x}_i)$ of the current node $n_1$ in the type derivation tree, upon entering its corresponding subtree, we construct the synthesis derivation tree for the $i$-th subgoal $P_i(\sigma_{i-1} \ldots \sigma_0 \sigma(\overline{x}_i))$ too.

When we return to the node $n_1$, the synthesis derivation trees for all subgoals have been constructed, allowing us to build node $n_2$ in the synthesis derivation tree. At every stage of the synthesis process, there is a node in the type derivation tree corresponding to the current synthesis goal.

\begin{itemize}
  \item \textbf{Selection of synthesis rules:}  
  Assume the node $n_1$ is labeled with a typing judgment $P(\overline{x}_c)$, annotated with a typing rule $R$, and the synthesis goal is $P(\overline{x})$. Following the procedure in \autoref{sec:meta-translation}, we translate $R$ into a synthesis rule $R'$, which is the synthesis rule to be applied.
  
  \item \textbf{Assignment of free variables:}  
    Assume the typing rule $R$ has the form in \autoref{T-Rule}, and the synthesis rule $R'$ has the form in \autoref{S-Rule}. 
    Then $P(\overline{x}_c)$ is the concretization of both the conclusion $P(\overline{x}_0)$ of the typing rule and the synthesis goal $P(\overline{x})$ (can be proved later). 
    Thus, $\overline{x}_0$ and $\overline{x}$ are unifiable, and the MGU $\sigma$ makes $\sigma(\overline{x}) = \sigma(\overline{x}_0)$ unifiable with $\overline{x}_c$. Therefore, there exists an assignment $\theta'$ such that $\theta'\sigma(\overline{x}_0) = \overline{x}_c$ and $\theta'$ covers $\mvhzc(\sigma(\overline{x}_0))$. Since $\mvhzc(\overline{x}_f) \subseteq \mvhzc(\overline{x}_0)$, we have $\mvhzc(\sigma(\overline{x}_f)) \subseteq \mvhzc(\sigma(\overline{x}_0))$, so $\theta'$ can also cover $\mvhzc(\sigma(\overline{x}_f))$. Let $\sigma_0 = \theta'|_{\mvhzc(\sigma(\overline{x}_f))}$ be the assignment for the free variables in the synthesis rule.
\end{itemize}

Since the translation begins from a valid type derivation tree, it is guaranteed to produce a valid synthesis derivation tree. In particular, the unification, assignment acquisition, and constraint checking steps will always succeed, and the synthesized assignment $\theta$ concretizes the original synthesis goal $P(\overline{x})$ to the typing judgment $P(\overline{x}_c)$ labeled at the root of the type derivation tree.

\begin{lemma}
\label{lemma:translation}
Let $T_1$ be a type derivation tree whose root node is labeled $P(\overline{x}_c)$, annotated with typing rule $R$, and $T_2$ be the translated synthesis derivation tree, whose root node is labeled $P(\overline{x}) \leadsto \theta$.
If $\,\overline{x}$ is unifiable with $\overline{x}_c$, then in the application of the synthesis rule $R'$ that derives $P(\overline{x}) \leadsto \theta$, (1) the unification step, (2) the acquisition step, and (3) the constraint checking step will succeed, and (4) the synthesized assignment $\theta$ concretizes $P(\overline{x})$ to $P(\theta(\overline{x}))$, which is equal to $P(\overline{x}_c)$.
\end{lemma}
\begin{proof}
  The complete proof can be found in the Appendix~\ref{sec:appendix-lemma3}.
\end{proof}

With Lemma~\ref{lemma:translation}, we can conclude that given a type derivation tree $T_1$ whose root node is labeled $\welltyped{p_c}$, we can start from the synthesis goal $\welltyped{\sk{p}}$, where $\sk{p}$ is a variable representing the program to be synthesized, and construct a synthesis derivation tree $T_2$ by translating from $T_1$ such that the root node of $T_2$ is labeled $\welltyped{\sk{p}} \leadsto \theta$, where $\theta(\sk{p}) = p_c$.

\subsection{Isomorphism between Type Derivation Trees and Synthesis Derivation Trees}
\label{sec:isomorphism}
Since every synthesis rule is translated from a typing rule, and every synthesis subgoal in the synthesis rule corresponds to a premise in the typing rule, we can establish an isomorphism between synthesis derivation trees and type derivation trees.

\paragraph{Isomorphism of Derivation Trees}
A type derivation tree $T_1 = (V_1, r_1)$ (with node set $V_1$ and root $r_1$) and a synthesis derivation tree $T_2 = (V_2, r_2)$ are \emph{isomorphic} if there exists a bijection $\varphi: V_1 \to V_2$ satisfying the following conditions:
\begin{enumerate}
    \item Root Correspondence: $\varphi(r_1) = r_2$
    \item Child Order Preservation: For $v \in V_1$, if $\mathsf{children}_{T_1}(v) = [v_1, \ldots, v_k]$, where $\mathsf{children}_{T}(v)$ denotes the ordered list of $v$'s children in tree $T$, then $\mathsf{children}_{T_2}(\varphi(v)) = [\varphi(v_1), \ldots, \varphi(v_k)]$
    \item Label/Annotation Correspondence: For every $v \in V_1$ labeled with typing judgment $P(\overline{t}_c)$ and a typing rule, $\varphi(v)$ is labeled with synthesis judgment $P(\skhzc{\overline{t}}) \leadsto \theta$, annotated with the corresponding synthesis rule and assignment, where $\theta(\skhzc{\overline{t}}) = \overline{t}_c$
\end{enumerate}
We write $T_1 \cong T_2$ to denote that $T_1$ and $T_2$ are isomorphic, representing the same logical structure but in different forms.
The translation procedure described in \autoref{sec:translation} provides a systematic way to construct an isomorphic synthesis derivation tree from any given type derivation tree. During this process, for each node $n_1$ in the type derivation tree, we define the correspondence $\varphi(n_1) = n_2$, where $n_2$ is the node constructed by translating $n_1$. It is straightforward to verify that this bijection preserves the root nodes and the order of children, and, by Lemma~\ref{lemma:translation}, also preserves the correspondence of labels and annotations. Thus, $T_1 \cong T_2$, as captured by the following lemma.

\begin{lemma}
\label{lemma:isomorphism1}
For any type derivation tree $T_1$, there exists a synthesis derivation tree $T_2$, $T_1 \cong T_2$.
\end{lemma}

Conversely, we can construct a type derivation tree $T_1$ from a given synthesis derivation tree $T_2$, also preserving isomorphism. The construction is straightforward: we traverse $T_2$, and for each node labeled with the synthesis goal $P(\overline{t}) \leadsto \theta$, we create a corresponding node labeled with the typing judgment $P(\theta(\overline{t}))$ and annotate it with the appropriate typing rule. Similar to Lemma~\ref{lemma:translation}, we can inductively show that the resulting $T_1$ is a type derivation tree, and that the bijection between $T_1$ and $T_2$ ensures the isomorphism $T_1 \cong T_2$.

\begin{lemma}
\label{lemma:isomorphism2}
For any synthesis derivation tree $T_2$, there exists a type derivation tree $T_1$, $T_1 \cong T_2$.
\end{lemma}
\begin{proof}
  The complete proof can be found in the Appendix~\ref{sec:appendix-lemma4}.
\end{proof}

\subsection{Soundness and Completeness}
\label{sec:soundness-completeness}
We have established an isomorphism between type derivation trees (type correctness proofs) and synthesis derivation trees (existential type correctness proofs). 
Under the isomorphism, \mainname ensures that well-typed programs can always be \emph{extracted} from synthesis derivation trees, rather than relying on direct program synthesis without reference to their type correctness proofs.

The notions of ``extract'' and ``well-typed'' can be formally defined in the context of derivation trees. 
Given a synthesis derivation tree with root node labeled $\welltyped{\sk{p}} \leadsto \theta$, the extracted program is $\theta(\sk{p})$. And a program $p_c$ is said to be well-typed if there exists a type derivation tree (type correctness proof) whose root node is labeled $\welltyped{p_c}$.

\begin{theorem}[Soundness]
Any program extracted from a synthesis derivation tree $T$ is well-typed.
\end{theorem}
\begin{proof}
Suppose the root node of $T$ is labeled $\welltyped{\sk{p}} \leadsto \theta$. By Lemma~\ref{lemma:isomorphism2}, there is a type derivation tree $T'$ such that $T \cong T'$. The root node of $T'$ is labeled $\welltyped{p_c}$, where $p_c = \theta(\sk{p})$ is the program extracted from $T$. By the definition of well-typed programs, $p_c$ is well-typed.
\end{proof}

\begin{theorem}[Completeness]
For any well-typed program $p_c$, there exists a synthesis derivation tree such that the program extracted from it is $p_c$.
\end{theorem}
\begin{proof}
There exists a type derivation tree $T_1$ whose root node is labeled $\welltyped{p_c}$. By Lemma~\ref{lemma:isomorphism1}, there exists a synthesis derivation tree $T_2$ such that $T_1 \cong T_2$. The root node of $T_2$ is labeled $\welltyped{\sk{p}} \leadsto \theta$, where $\theta(\sk{p}) = p_c$. Thus, the program extracted from $T_2$ is $p_c$.
\end{proof}

\begin{corollary}
The synthesis derivation tree $T$ guaranteed by the completeness theorem is precisely the tree obtained through the translation process described in \autoref{sec:translation}.
\end{corollary}

The soundness and completeness theorems ensure the correspondence between well-typed programs and synthesis derivation trees. This correspondence works in both directions: every well-typed program can be represented as a synthesis derivation tree (completeness), and every synthesis derivation tree yields a well-typed program (soundness). Together, these results confirm that constructing a synthesis derivation tree is equivalent to constructing an existential type correctness proof: the isomorphism with type derivation trees provides the proof structure, while the synthesized assignment provides the witness (the program) that makes the existential statement $\exists p.\; \welltyped{p}$ hold.
\section{Language Model Design}
\label{sec:implementation}

\subsection{Model Architecture}
To better adapt to our proof-guided synthesis system and maximize token utilization efficiency, we design a novel encoder-decoder-based model architecture, illustrated in \autoref{fig:model-arch}.

Unlike standard encoder-decoder models that encode only a static input, our architecture features a \emph{dual-encoding} design that separately processes two distinct input streams: (1) the fixed natural language specification, and (2) the dynamically evolving synthesis goal at each step.

\paragraph{Static Encoding.} The natural language specification remains unchanged throughout the synthesis process. We encode it once at the beginning to produce a fixed sequence of hidden states that stays constant across all synthesis steps.

\paragraph{Dynamic Encoding.} The synthesis goal, which provides the \emph{dynamic typing context}, evolves at each step as the synthesis derivation tree grows. We re-encode the current synthesis goal at each step to capture the updated typing context. The encoder produces a separate sequence of hidden states for this dynamic input.

\paragraph{Cross-Attention Fusion.} The decoder accesses both encodings through cross-attention mechanisms. We concatenate the static encoding $H_s$ and dynamic encoding $H_d$ along the sequence dimension to form a unified key-value sequence $H = [H_s; H_d]$. The decoder's cross-attention layers then attend to this concatenated sequence, computing attention over both the task specification and the current typing context simultaneously. This design allows the model to jointly reason about both information sources when predicting the next synthesis decision.

\begin{figure}[t]
\centering
\includegraphics[width=0.9\textwidth]{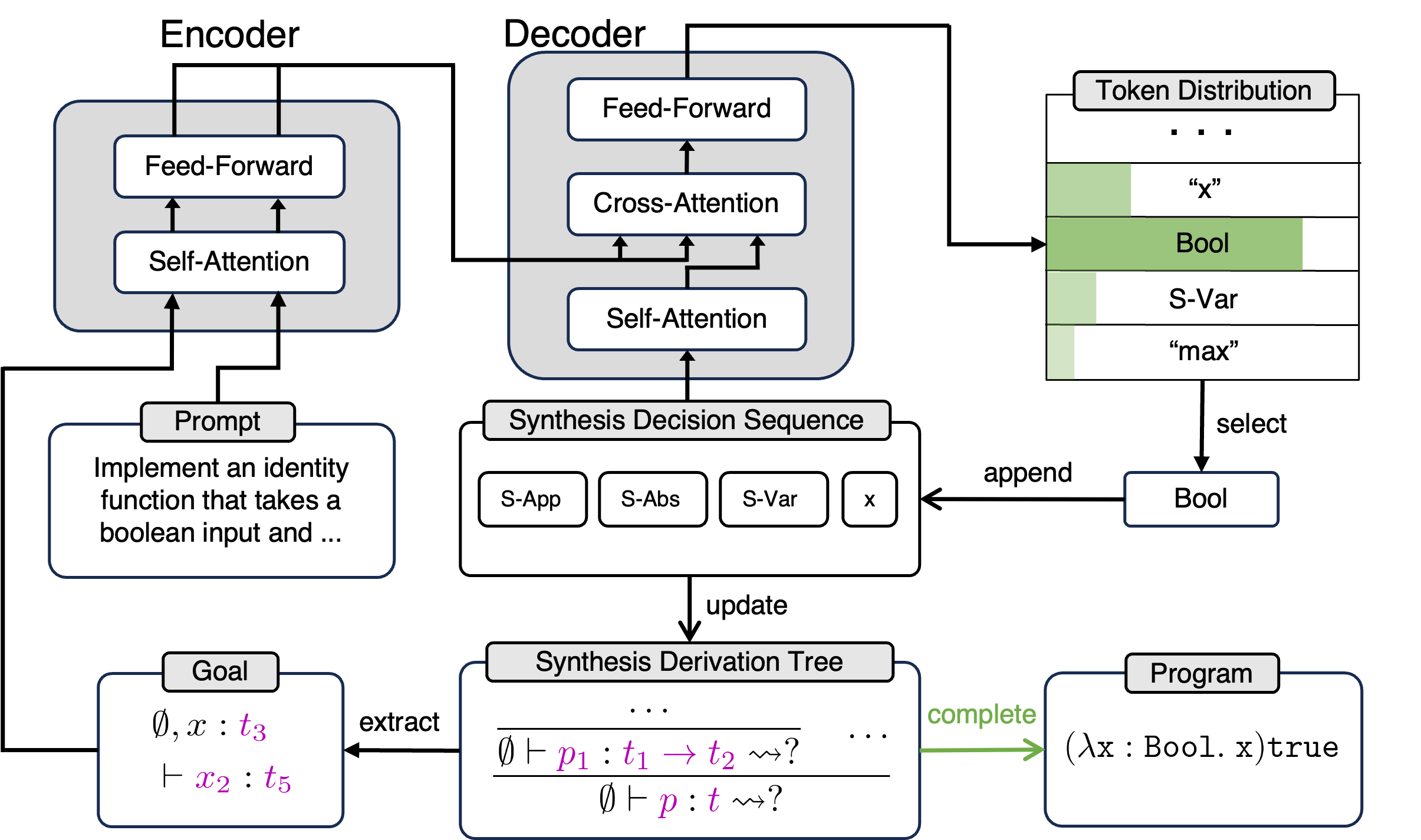}
\caption{Model Architecture for Constructing Synthesis Derivation Trees}
\label{fig:model-arch}
\end{figure}

As described in \autoref{sec:query-resolution-lm}, the query resolution task is delegated to the LM, which takes the natural language prompt, the synthesis decisions generated so far, and the current synthesis goal as input and outputs a token representing the synthesis rule or variable assignments.

Meanwhile, \mainname maintains a synthesis decision sequence where each newly generated token is appended, representing the next synthesis decision. This sequence incrementally constructs an existential type correctness proof. Due to its autoregressive nature, the growing sequence is fed to the transformer decoder, allowing the LM to leverage the full history of synthesis decisions when predicting the next token.

\paragraph{Model Inference Process.}
At each synthesis step:
\begin{enumerate}
    \item The current synthesis goal is extracted from the (incomplete) synthesis derivation tree.
    \item The encoder encodes the synthesis goal to provide \emph{dynamic typing context} to the LM.
    \item The decoder processes both encodings from the encoder, along with the previous synthesis decision sequence, generating the next token representing a synthesis decision.
    \item This new token is appended to the synthesis decision sequence and used to incrementally update the synthesis derivation tree as detailed in  \autoref{alg:synthtree}.
    \item The updated synthesis derivation tree yields either a new synthesis goal for the next iteration or, upon completion, a well-typed program together with its type correctness proof.
\end{enumerate}

\paragraph{Efficient Token Reuse through Incremental Construction.}
According to \autoref{alg:synthtree}, the synthesis derivation tree construction follows a serialized order, ensuring that each step only appends new decisions without modifying the existing sequence.
This property aligns naturally with the autoregressive generation paradigm of transformer decoders, where each new token is predicted based on all previously generated tokens.
As a result, we can reuse the transformer's key-value (KV) cache~\cite{DBLP:conf/mlsys/PopeDCDBHXAD23}: the intermediate attention states computed for earlier tokens in the decision sequence are cached and reused when generating subsequent tokens. This eliminates the need to recompute attention over the entire decision sequence at each decision-making step, significantly reducing the computational overhead from quadratic to linear in sequence length.

\paragraph{Model Configuration.}
Our dual-encoding architecture is built upon standard Transformer encoder-decoder models, where both the encoder and decoder follow the standard Transformer architecture~\cite{DBLP:conf/nips/VaswaniSPUJGKP17}. We evaluate on two pre-trained models: \textbf{CodeT5-220M}~\cite{DBLP:conf/emnlp/0034WJH21}, with 12-layer encoder and 12-layer decoder, 12 attention heads, a hidden dimension of 768, and approximately 220 million parameters; and \textbf{T5Gemma2-2B}~\cite{zhang2025t5gemma2seeingreading}, with 26-layer encoder and 26-layer decoder, 8 attention heads, a hidden dimension of 2304, and approximately 2 billion parameters. For both models, the static encoder and dynamic encoder share the same weights but process different input streams.

\subsection{Term Encoding}
\label{sec:term-encoding}
Terms, as defined in \autoref{terms-intro}, are inductively constructed from constants and constructors. Following the approach of GrammarT5~\cite{DBLP:conf/icse/ZhuL000C24}, we serialize the abstract syntax tree (AST) of a term using a pre-order traversal, outputting each constructor or constant as a token in sequence. When a term contains unknown variables (not yet instantiated during synthesis), we represent them using a unified unknown token $\langle?\rangle$.

For example, consider the STLC program $\tcode{abs}(\tcode{"x"}, \tcode{bool}, \tcode{var}(\tcode{"x"}))$, which represents the identity function $\lambda x{:}\tcode{bool}.\, x$. Its AST has the constructor $\tcode{abs}$ at the root, with three children: the constant $\tcode{"x"}$ (the variable name), the constant $\tcode{bool}$ (the type annotation), and the subtree $\tcode{var}(\tcode{"x"})$. The pre-order traversal yields the token sequence illustrated in \autoref{fig:term-encoding}.

\begin{figure}[tb]
\centering
\begin{tikzpicture}[x=1.2cm, y=0.8cm]
    \tikzset{
        ctor/.style={
            draw=blue!60, fill=blue!15,
            minimum height=5mm, minimum width=8mm,
            rounded corners=1.5pt, font=\small,
            text depth=0.25ex, text height=1.6ex
        },
        const/.style={
            draw=blue!40, fill=blue!8,
            minimum height=5mm, minimum width=8mm,
            rounded corners=1.5pt, font=\small,
            text depth=0.25ex, text height=1.6ex
        },
        unknown/.style={
            draw=blue!30, fill=blue!5,
            minimum height=5mm, minimum width=8mm,
            rounded corners=1.5pt, font=\small,
            text depth=0.25ex, text height=1.6ex,
            dashed
        },
    }

    \node at (-1.5, 0) {\small (a)};
    \node[ctor] (c1) at (0,0) {\tcode{abs}};
    \node[const] (c2) at (1,0) {\tcode{"x"}};
    \node[const] (c3) at (2,0) {\tcode{bool}};
    \node[ctor] (c4) at (3,0) {\tcode{var}};
    \node[const] (c5) at (4,0) {\tcode{"x"}};

    \node at (-1.5, -1.2) {\small (b)};
    \node[ctor] (p1) at (0,-1.2) {\tcode{abs}};
    \node[const] (p2) at (1,-1.2) {\tcode{"x"}};
    \node[const] (p3) at (2,-1.2) {\tcode{bool}};
    \node[unknown] (p4) at (3,-1.2) {$\langle?\rangle$};

\end{tikzpicture}

\setlength{\fboxsep}{1pt}
\caption{Term encoding as token sequences: 
\colorbox{blue!15}{\textsf{constructors}} and 
\colorbox{blue!8}{\textsf{constants}}.
(a) Full term $\tcode{abs}(\tcode{"x"}, \tcode{bool}, \tcode{var}(\tcode{"x"}))$.
(b) Partial term $\tcode{abs}(\tcode{"x"}, \tcode{bool}, \sk{e})$ with unknown variable.}
\label{fig:term-encoding}
\end{figure}

This encoding preserves the structural information of the term and can be uniquely decoded back to the original AST. The decoding is possible because each constructor has a fixed arity and known parameter types (as specified in the term definitions in \autoref{terms-intro}). When the decoder encounters a constructor token, it knows exactly how many subsequent tokens to consume and what types they should have, enabling unambiguous reconstruction of the tree structure. If the number of arguments or their types do not match the constructor's specification, the decoding fails and reports an error. During synthesis, the LM generates tokens following this same pre-order traversal order to construct terms for variable assignments. As each token is generated, the synthesis system incrementally decodes the sequence to reconstruct the corresponding AST, enabling immediate validation and integration into the synthesis derivation tree.

\subsection{Advanced Techniques for LM-Based Program Synthesis}
The program synthesis algorithm in \autoref{alg:synthtree} involves sequential actions at each node: selecting synthesis rules, acquiring variable assignments, and applying rules to update the synthesis derivation tree. 
To ensure the efficiency and correctness of this process, we implement pruning strategies that naturally correspond to the potential failures in synthesis rule application. 

At the syntactic level, we apply \emph{syntactic pruning}, which ensures the structural validity of synthesis derivation trees and enforces the well-formedness of generated sequences.
\emph{Syntactic pruning} is triggered instantly when a token is generated: (1) the synthesis goal expects a synthesis rule to apply but the token stands for a term, or vice versa; (2) the predicate in the synthesis rule's conclusion differs from that of the synthesis goal; or (3) during term generation, the token corresponds to constructors or constants (see \autoref{terms-intro}) that mismatch the expected term type.
\begin{example}
Consider generating the term $\tcode{abs}(\tcode{"x"}, \tcode{bool}, \tcode{var}(\tcode{"x"}))$ using the term encoding from \autoref{sec:term-encoding}. Suppose the current token sequence is $[\tcode{abs}, \tcode{"x"}]$, meaning the LM has generated the constructor $\tcode{abs}$ and its first argument $\tcode{"x"}$ (the variable name). According to the constructor specification $\tcode{abs} : (\tcode{String}, \tcode{Type}, \tcode{Prog})$, the next token should be of type $\tcode{Type}$. If the LM generates $\tcode{var}$ (a $\tcode{Prog}$ constructor) instead of a valid type like $\tcode{bool}$, syntactic pruning immediately rejects this token because the expected type $\tcode{Type}$ does not match the actual type $\tcode{Prog}$.
\end{example}

At the type level, we apply \emph{type pruning}, which focuses on type-level correctness and relies on premise checking during the application of synthesis rules. 
If the unification step or constraint checking step fails, the corresponding search is immediately terminated. 
This prevents the exploration of synthesis paths that cannot lead to valid type correctness proofs.
\begin{example}
Consider applying the \rulename{T-Var} rule, which has a constraint $\checkBinding{\Gamma}{x}{t}$ that checks whether the variable $x$ with type $t$ exists in the context $\Gamma$. Suppose the current synthesis goal is $\stlctype{\Gamma}{\sk{e}}{\tcode{int}}$ where $\Gamma = \{f : \tcode{bool}, g : \tcode{int}\}$. If the LM selects $f$ to instantiate $\sk{e}$, the constraint checking step verifies $\checkBinding{\Gamma}{f}{\tcode{int}}$. Since $f$ has type $\tcode{bool}$ in $\Gamma$, not $\tcode{int}$, the constraint fails and type pruning terminates this branch. The correct choice would be $g$, which has the matching type $\tcode{int}$.
\end{example}

\emph{Syntactic pruning}, \emph{type pruning}, and the previously mentioned \emph{dynamic typing context} constitute a layered enhancement framework: 
\emph{syntactic pruning} serves as a prerequisite for \emph{type pruning}, since a structurally invalid synthesis derivation cannot be applied to perform type-level checking. 
Furthermore, without the successful application of synthesis rules, the intended synthesis goal cannot be extracted from the synthesis derivation tree. 
Should all three techniques be omitted, the model reduces to a standard encoder-decoder architecture that directly maps natural language inputs to sequences of synthesis decisions, lacking any mechanism to enforce structural or type-level correctness.

Beam search~\cite{DBLP:conf/aclnmt/FreitagA17} enables efficient exploration of multiple synthesis paths by maintaining the top-$k$ candidate sequences based on cumulative likelihood scores. At each step, rather than greedily selecting only the highest-probability token, beam search expands all $k$ beams and retains the top-$k$ sequences across all candidates.
When a branch is pruned due to syntactic or type errors, we terminate its sequence and discard its KV cache, while remaining branches continue independently. This prevents invalid branches from wasting beam slots, ensuring all $k$ slots are reserved for branches that may lead to valid programs.
\section{Evaluation \label{sec:evaluation}}
We design our evaluation to answer the following research questions.
\begin{itemize}
    \item \textbf{RQ1:} How does \mainname compare with the existing code generation methods?
    \item \textbf{RQ2:} How does each component of \mainname contributes to the overall performance?
    \item \textbf{RQ3:} How does the correctness of \mainname compare with rejection sampling, the previous method for ensuring the correctness of code generation?
    \item \textbf{RQ4:} How does \mainname compare with variants that generate the type derivation tree and the program separately instead of simultaneously?
\end{itemize}

\subsection{Experimental Setup}

\subsubsection{Benchmark}
To evaluate the generality of \mainname, our evaluation considers two languages with very different characteristics, SuFu~\cite{DBLP:journals/pacmpl/JiZPXH24} and Java.
\begin{itemize}
 \item \textbf{SuFu} is an ML-style functional language designed for automatic program optimization. This language uses a customized type system for extracting operations for optimization on data structures.
 This language has limited code resources, making it representative for evaluating \mainname under low-resource conditions.

 We collect SuFu programs from its GitHub repository, which comprises 290 programs with test cases. For each program, we manually supply a natural language description, along with the relevant inductive definitions and library functions, thus collecting a dataset of code generation tasks.
 This dataset is then split randomly into a training set of 80\% tasks and a test set of 20\% tasks.
 
 \item \textbf{Java} is a popular imperative language with a sophisticated type system. In a balance between the implementation cost and the expressiveness, we implement a subset of Java that preserves its core syntax while omitting advanced features like lambda expressions.

 We consider the MBJP dataset in our evaluation, which is the Java version of the popular MBPP dataset~\cite{DBLP:journals/corr/abs-2107-03374}. Our Java subset covers 608 tasks from MBJP, and we split these tasks randomly into a training set of 90\% tasks and a test set of 10\%. 
\end{itemize}

More information about SuFu, Java, and the datasets can be found in Appendix~\ref{appendix-benchmark}.

\subsubsection{Models}
We implement \mainname based on CodeT5~\cite{DBLP:conf/emnlp/0034WJH21}, a family of encoder-decoder models for code understanding and generation. We adopt the 220M variant of CodeT5 as our base model for subsequent training and refer to our model as \mainname-220M. 
We also evaluate \mainname on T5Gemma2~\cite{zhang2025t5gemma2seeingreading}, a more recent encoder-decoder model family built upon Gemma 3. We adopt the 2B variant of T5Gemma2 as our base model and refer to this model as \mainname-2B. 

\subsubsection{Evaluation Metrics} 
\label{sec:evaluation-metrics}
To evaluate the quality of code generation, we adopt the metrics below, which respectively measure the correctness, effectiveness, and type correctness:

\begin{itemize}
    \item \textbf{Pass@k (k=1,10)~\cite{DBLP:journals/corr/abs-2107-03374}}: These metrics quantify the proportion of problems where at least one of the top-k generated candidates passes all reference test cases. Each generated candidate is executed against the test cases provided in the benchmark, and a candidate is considered correct only if it produces the expected output for all test cases. 

    \item \textbf{First Success Position (FSP)}: This metric is the mean rank of the first functionally correct candidate among the generated ones across the test set. It measures how early in the candidate list a correct solution appears, assessing the efficiency of output ranking.

    \item \textbf{Compilation Error Rate (CER)}: This metric quantifies the proportion of generated candidates that fail to compile. 
    Since \mainname generates programs with type correctness guarantees, CER reflects the type system's effectiveness in catching compile-time errors. 
\end{itemize}

\subsubsection{Implementation Details}

We fine-tuned both \mainname-220M (based on CodeT5-220M) and \mainname-2B (based on T5Gemma2-2B) following a standard procedure with the same set of hyperparameters to ensure a fair comparison. Each model was fine-tuned on the training set until convergence, with multiple checkpoints saved throughout the process. The best-performing checkpoint is used for evaluation.
Since \mainname introduces a new representation that differs from standard code tokens, we perform an initial fine-tuning phase using a supervised fine-tuning dataset from OpenCoder~\cite{DBLP:journals/corr/abs-2411-04905} to help the model adapt to this representation.

All experiments were performed on a computing server equipped with 2~$\times$~AMD EPYC 9655 CPUs, 512~GB DDR5 RAM, and 8~$\times$~NVIDIA GeForce RTX 4090 GPUs, running Ubuntu 22.04.5 LTS with CUDA 12.8.

\subsection{Experimental Results}

\subsubsection{RQ1: Performance of \mainname}
\autoref{tab:model-results} summarizes the overall performance of \mainname. 
The top four models are fine-tuned and evaluated on the SuFu dataset, while the bottom four models are on the MBJP dataset. For each language, we compare the baseline models (CodeT5-220M and T5Gemma2-2B) with their corresponding \mainname variants.

\begin{table}[ht]
\centering
\begin{threeparttable}
\caption{Comparison of Accuracy and Compilation Error Rate between CodeT5 and \mainname on SuFu and Java Datasets}
\label{tab:model-results}
\begin{tabular}{
    l 
    c
    S[table-format=2.2] 
    S[table-format=2.2] 
    S[table-format=1.2] 
    S[table-format=2.2]
}
\toprule
Language & Model & {pass@1 (\%)\tnote{\updir}} & {pass@10 (\%)\tnote{\updir}} & {FSP\tnote{\downdir}} & {CER (\%)\tnote{\downdir}} \\ 
\midrule
\multirow{4}{*}{SuFu} 
        & CodeT5-220M         & 24.14 & 32.76 & 7.03 & 83.10 \\
        & T5Gemma2-2B         & 29.31 & 37.93 & 6.69 & 61.21 \\
        \cmidrule{2-6}
        & \mainname-220M      & 37.93 & 46.55 & 5.48 & 0.00 \\
        & \mainname-2B        & \textbf{43.10} & \textbf{50.00} & \textbf{5.03} & \textbf{0.00} \\
\midrule
\multirow{4}{*}{Java} 
        & CodeT5-220M         & 10.45 & 20.90 & 8.19 & 38.51  \\
        & T5Gemma2-2B         & 17.91 & 35.82 & 6.99 & 15.22 \\
        \cmidrule{2-6}
        & \mainname-220M      & 11.94 & 28.36 & 7.94 & 3.52 \\
        & \mainname-2B        & \textbf{23.19} & \textbf{40.30} & \textbf{6.76} & \textbf{3.12} \\
\bottomrule
\end{tabular}
\begin{tablenotes}
\footnotesize
\item \updir\, Higher is better; \downdir\, Lower is better; \textbf{bold} indicates best value per column.
\item pass@k measures the proportion of problems solved by at least one of the top-k candidates (verified by test cases); FSP is the average rank of the first correct candidate; CER is the compilation error rate. See \autoref{sec:evaluation-metrics} for details.
\end{tablenotes}
\end{threeparttable}
\end{table}

\mainname consistently improves code generation accuracy across both base models. On SuFu, \mainname-220M improves pass@10 from 32.76\% (CodeT5-220M) to 46.55\%, and \mainname-2B improves it from 37.93\% (T5Gemma2-2B) to 50.00\%. On Java, \mainname-220M improves pass@10 from 20.90\% to 28.36\%, and \mainname-2B improves it from 35.82\% to 40.30\%.
The improvements are particularly pronounced on SuFu, indicating that \mainname enhances the model's ability to reason about type constraints on languages with complex type systems and limited training data.

In terms of the position of the first successful prediction (FSP), \mainname consistently reduces FSP compared to the corresponding base models. On SuFu, \mainname-220M reduces FSP from 7.03 to 5.48, and \mainname-2B reduces it from 6.69 to 5.03. On Java, \mainname-220M reduces FSP from 8.19 to 7.94, and \mainname-2B reduces it from 6.99 to 6.76.

Most notably, \mainname greatly reduces compilation errors, as indicated by the CER numbers.
In SuFu, all compilation errors are caused by type errors, so \mainname completely eliminates compilation errors (CER = 0.00\%). 
In Java, compilation errors may also be triggered by static analysis, which is uncaptured in our implemented type system(e.g., unreachable code), which is why \mainname retains a small CER.
Nonetheless, \mainname still reduces compilation errors significantly: \mainname-220M reduces CER from 38.51\% to 3.52\%, and \mainname-2B reduces it from 15.22\% to 3.12\%.

Overall, \mainname achieves clear performance gains and significant improvements in code correctness and reliability for both programming languages, fulfilling our design objectives.

\subsubsection{RQ2: Contributions of Individual Components in \mainname}
To answer RQ2, we conduct an ablation study on the SuFu benchmark and the \mainname-220M model to quantify the contributions of each component in \mainname. Starting from the base case where \mainname synthesizes decision sequences with no checks, system components (syntactic pruning, type pruning, and dynamic typing context, as introduced in \autoref{sec:implementation}) are incrementally added, and the impact on performance metrics is evaluated at each stage. 

\begin{table}[ht]
\centering
\setlength{\tabcolsep}{4pt}
\begin{threeparttable}
\caption{Sequential Ablation Study: Performance with Incremental Component Additions}
\label{tab:ablation-combined}
\begin{tabular}{
    l 
    S[table-format=2.2]@{\hspace{1em}}S[table-format=+2.2]
    S[table-format=2.2]@{\hspace{1em}}S[table-format=+2.2]
    S[table-format=1.2]@{\hspace{1em}}S[table-format=+1.2]
    S[table-format=2.2]@{\hspace{1em}}S[table-format=+2.2]
}
\toprule
\multicolumn{1}{l}{Components} &
\multicolumn{2}{c}{pass@1 (\%)\tnote{\updir}} &
\multicolumn{2}{c}{pass@10 (\%)\tnote{\updir}} &
\multicolumn{2}{c}{FSP\tnote{\downdir}} &
\multicolumn{2}{c}{CER (\%)\tnote{\downdir}} \\
\cmidrule(lr){2-3} \cmidrule(lr){4-5} \cmidrule(lr){6-7} \cmidrule{8-9}
 & {Val} & {$\Delta$} & {Val} & {$\Delta$} & {Val} & {$\Delta$} & {Val} & {$\Delta$} \\
\midrule
Base (no checks)         & 24.14 & 0.00  
                        & 36.21 & 0.00   
                        & 6.78  & 0.00  
                        & 74.26 & 0.00   \\
+ Syntactic Pruning         & 27.59 & +3.45  
                        & 36.21 & 0.00  
                        & 6.71  & -0.07 
                        & 72.13 & -2.13  \\
+ Type Pruning            & \textbf{37.93} & \textbf{+13.79} 
                        & 43.10 & +6.89 
                        & 5.79  & -0.99 
                        & \textbf{0.00} & \textbf{-74.26} \\
+ Dynamic Typing Context       & \textbf{37.93} & \textbf{+13.79}
                        & \textbf{46.55} & \textbf{+10.34}
                        & \textbf{5.48}  & \textbf{-1.30}
                        & \textbf{0.00}  & \textbf{-74.26} \\
\bottomrule
\end{tabular}
\begin{tablenotes}
\footnotesize
\item[\,] \textit{Val}: absolute value for each setting; $\Delta$: improvement relative to the base setting.
\end{tablenotes}
\end{threeparttable}
\end{table}

Starting from the base setting, each subsequent enhancement contributes meaningfully to the overall improvement in code generation quality and reliability. 

Syntactic pruning improves pass@1 by +3.45\% and reduces CER by 2.13\%, demonstrating that basic structural constraints help filter trivially invalid outputs. 
Type pruning delivers transformative results by eliminating all compilation errors (CER drops to 0.00\%) while substantially boosting pass@1 to 37.93\% (+13.79\%) and pass@10 to 43.10\% (+6.89\%), highlighting type correctness as critical for synthesis. 
Dynamic typing context further enhances performance, pushing pass@10 to 46.55\% (+10.34\%) and reducing FSP to 5.48 (-1.30), showing that richer typing context enables generating more correct solutions among higher-ranked outputs.

The results validate our motivation: the type-based mechanisms are essential for ensuring type correctness.
In particular, the effect of dynamic typing context demonstrates the importance of \textbf{Context Locality} in ~\autoref{section:intro}: by providing necessary type information within a compact local context, this technique significantly improves the effectiveness of generation.

\subsubsection{RQ3: Comparison with Rejection Sampling}
To answer RQ3, we compare \mainname with two baselines: the vanilla CodeT5-220M (without any post-processing) and CodeT5-220M with rejection sampling, the previous state-of-the-art for achieving guarantees in code generation. 
In our case, rejection sampling queries the LM repeatedly and discards those candidate programs that cannot be compiled until finding sufficient solutions or reaching a resource limit.
Notably, rejection sampling is functionally equivalent to constrained decoding, as both methods effectively re-normalize the model's output distribution over the subspace of valid, compilable programs. By including rejection sampling, we establish a strong baseline for the performance of the model under strict correctness constraints.
\begin{table}[ht]
\centering
\begin{threeparttable}
\caption{Comparison: Vanilla CodeT5, Rejection Sampling, and \mainname}
\label{tab:model-compare-sufu-java}
\begin{tabular}{
    l
    l
    S[table-format=2.2]
    S[table-format=2.2]
}
\toprule
Language & Model & {pass@1 (\%)\tnote{\updir}} & {pass@10 (\%)\tnote{\updir}} \\
\midrule
\multirow{3}{*}{SuFu} 
        & CodeT5                & 24.14 & 32.76 \\
        & CodeT5 + Rejection Sampling & 31.03 & 32.76 \\
        & \mainname-220M             & \textbf{37.93} & \textbf{46.55} \\
\midrule
\multirow{3}{*}{Java} 
        & CodeT5                & 10.45 & 20.90 \\
        & CodeT5 + Rejection Sampling & 10.45 & 20.90 \\
        & \mainname-220M             & \textbf{11.94} & \textbf{28.36} \\
\bottomrule
\end{tabular}
\end{threeparttable}
\end{table}

On both datasets, rejection sampling can provide only minor improvements over the baseline, much less than the improvements achieved by \mainname.
This demonstrates that rejection sampling, as a post-hoc filtering technique, fails to enhance the LM's ability to generate correct code and is thus limited by the candidate quality produced by the base LM. 
Similarly, constrained decoding, which also relies on filtering during generation, would face comparable limitations as it cannot fundamentally improve over the base LM, either.

In contrast, \mainname consistently achieves higher pass@1 and pass@10, reflecting its ability to generate more correct and top-ranked predictions. 
Unlike filtering-based methods, \mainname explicitly incorporates typing rules and typing context during training, enabling the model to learn and utilize type constraints. 
These results highlight that embedding type information into LMs is substantially more effective than post-hoc pruning approaches, demonstrating the critical importance of \textbf{Type Explicitness} in \autoref{section:intro}.

\subsubsection{RQ4: Comparison with Separated Type-Code Generation}
To answer RQ4, we compare \mainname with approaches that separate type reasoning and code generation. 
We evaluate three variants: (1) Type-First: CodeT5-220M is fine-tuned to generate a sequence of typing rules (forming the type derivation tree) as explicit type reasoning, then generate code based on the reasoning, (2) Code-First: CodeT5-220M is finetuned to first generate code, then output a sequence of typing rules as its type correctness proof, and (3) \mainname-220M: our synthesis decision sequence representation.

\begin{table}[ht]
\centering
\begin{threeparttable}
\caption{Comparison: Separated vs. Integrated Type-Code Generation}
\label{tab:sequential-compare}
\begin{tabular}{
    l
    l
    S[table-format=2.2]
    S[table-format=2.2]
    S[table-format=2.2]
    S[table-format=1.2]
    S[table-format=3.0]
}
\toprule
Language & Model & {pass@1 (\%)\tnote{\updir}} & {pass@10 (\%)\tnote{\updir}} & {CER (\%)\tnote{\downdir}} & {FSP\tnote{\downdir}} & {Avg Tokens\tnote{\downdir}} \\
\midrule
\multirow{3}{*}{SuFu}
        & CodeT5 + Type-First        & 31.03 & 36.21 & 64.31 & 6.57 & 752 \\
        & CodeT5 + Code-First        & 29.31 & 39.66 & 56.72 & 6.28 & 752 \\
        & \mainname-220M             & \textbf{37.93} & \textbf{46.55} & \textbf{0.00} & \textbf{5.48} & \textbf{410} \\
\midrule
\multirow{3}{*}{Java}
        & CodeT5 + Type-First        & \textbf{11.94} & 26.87 & 37.61 & \textbf{7.36} & 247 \\
        & CodeT5 + Code-First        & 8.96 & 22.39 & 31.49 & 8.21 & 247 \\
        & \mainname-220M             & \textbf{11.94} & \textbf{28.36} & \textbf{3.52} & 7.94 & \textbf{129} \\
\bottomrule
\end{tabular}
\end{threeparttable}
\end{table}

\autoref{tab:sequential-compare} compares separated reasoning approaches against \mainname on both SuFu and Java benchmarks.
Both sequential variants achieve pass@10 of 36--40\% on SuFu and 22--27\% on Java, while \mainname achieves 46.55\% and 28.36\% respectively.
Moreover, they still suffer from high compilation error rates (57--64\% on SuFu and 31--38\% on Java), whereas \mainname eliminates compilation errors on SuFu and reduces them to 3.52\% on Java.

Besides the correctness advantages, \mainname also outperforms both baselines in token usage.
In comparison, \mainname reduces token consumption by 45\% on SuFu (410 vs. 752 tokens) and 48\% on Java (129 vs. 247 tokens) compared to the baselines.

These results validate \mainname's \textbf{Derivation Vicinality} paradigm: by unifying type reasoning and code generation into a single synthesis process, \mainname achieves both higher accuracy and improved efficiency.

\section{Related Work}\label{sec:related}

\subsection{Rejecting Invalid Generations}
Existing studies have primarily explored \emph{external} enforcement of symbolic constraints into LMs.
A common baseline is the \emph{generate-and-validate} approach~\cite{DBLP:conf/sigsoft/ZhuSXZY0Z21,DBLP:conf/popl/LongR16}, which repeatedly generates code and checks it with an external validator (e.g., a compiler) until a valid candidate is found. While this ensures constraint satisfaction, it is often inefficient due to a high rate of invalid generations.

To address this, \emph{constrained decoding} techniques perform on-the-fly validation during token generation, discarding tokens that cannot lead to a valid complete program. This proactive strategy improves efficiency and has been applied to enforce various constraints, including syntactic correctness~\cite{ugare2024syncodellmgenerationgrammar}, type correctness~\cite{DBLP:journals/corr/abs-2504-09246}, and vulnerability avoidance~\cite{DBLP:conf/issre/StorhaugLH23}. However, such token-level pruning can also distort the LM's output distribution~\cite{DBLP:conf/nips/ParkWBPD24}. We compare with rejection sampling in our evaluation (\autoref{tab:model-compare-sufu-java}), which is functionally equivalent to constrained decoding as both re-normalize the model's output distribution over compilable programs.

Critically, both strategies enforce constraints \emph{externally}. As discussed, externally enforced constraints do not modify the LM's internal reasoning process. Consequently, the model's internal distribution may become skewed, potentially assigning low probability to correct code snippets due to inaccurate assessments of their typability.

\subsection{Internalizing Symbolic Constraints}
Another line of research aims to expose the reasoning process of symbolic constraints to the LMs, thereby guiding them to learn to satisfy the constraints.
By externalizing the complexity of constraint reasoning, this paradigm relieves LMs from implicitly inferring intricate symbolic rules from plain text. Offloading constraint learning to the representation allows LMs to allocate more capacity to other critical aspects, such as program semantics and algorithmic logic. This reallocation of modeling resources leads to improved generation quality that scales effectively with model size, mirroring the benefits observed in syntactic representations for grammatical learning~\cite{DBLP:conf/acl/LiangZ0LLXCZZZC25}.

Grammar-based encoding~\cite{DBLP:conf/acl/YinN17,DBLP:conf/aaai/SunZXSMZ20,DBLP:conf/icse/ZhuL000C24} represents one such approach, where programs are parsed into ASTs and serialized into sequences of grammar rules for training and inference. 
Grammar-based encoding can be viewed as a special case of our framework when grammar rules are expressed as CHCs. Our approach generalizes this concept to arbitrary CHC-representable constraints.

Tare~\cite{DBLP:conf/icse/ZhuSZXZ23} introduces a neural approach for learning Java's type system in program repair, using specialized components to encode typing contexts and subtyping relations. While it predicts both code and types, Tare is specifically tailored to a subset of Java's type system and operates through AST traversal without achieving full type explicitness, which is a key distinction from our work.

Refine4LLM~\cite{DBLP:journals/pacmpl/CaiHSLLSD25} generates programs through sequential refinement rule applications, transforming specifications into concrete programs while ensuring syntactic and functional correctness. However, it requires training data containing refinement processes, which are scarce in practice, and is specific to a particular refinement calculus. In contrast, our approach utilizes readily available programs for training given a type checker, and supports any constraints expressible as CHCs.

\subsection{Type-Guided Program Synthesis}
Our approach is inspired by type-guided program synthesis~\cite{DBLP:journals/pacmpl/GuoJJZWJP20, DBLP:conf/pldi/GuriaFH21, DBLP:conf/pldi/PolikarpovaKS16, DBLP:conf/pldi/OseraZ15}, which employs synthesis rules to produce type-correct programs. These approaches also guarantee type correctness by construction; however, they are designed for classical enumerative or constraint-based synthesis algorithms, not for neural code generation.

In contrast, our work addresses three challenges that arise when adapting type-guided synthesis to LMs.
First, existing synthesis systems do not provide a code representation suitable for LM training. We design synthesis decision sequences that can be derived from existing programs via type derivation trees, enabling the use of readily available codebases as training data.
Second, existing systems do not enforce the structural correspondence between type derivations and synthesis derivations. We design our synthesis rules to satisfy Type Explicitness and Derivation Vicinality, establishing an isomorphism that makes type reasoning fully explicit in the generated sequences.
Third, existing systems do not consider neural architectures for learning type systems. We propose a dual-encoding architecture that enables LMs to effectively learn and utilize type constraints during generation.

\section{Conclusion \label{sec:conclusion}}
This paper presents \mainname, a proof-guided approach that guarantees type correctness in code generation by constructing existential type correctness proofs through a synthesis process. Through a synthesis system isomorphic to the type system and a synthesis decision sequence representation, \mainname ensures that constructing a synthesis derivation tree is equivalent to constructing an existential type correctness proof, thereby eliminating type errors during code generation, while significantly improving overall generation performance. 

Beyond type correctness, \mainname's foundation on CHC enables broader applications, including safety verification and the generation of other structural data. Our core insight---that aligning program representations with underlying proof systems enhances the generation performance---opens promising directions to other domains where structural constraints matter.


\bibliographystyle{ACM-Reference-Format}
\bibliography{bibtex}

\newpage
\appendix
\section{Proofs for Some Lemmas Used in the Paper}

\subsection{Proof for Lemma~\ref{rule-bijection}}
\label{sec:appendix-lemma1}
\begin{lemma}
There is a bijection $\varphi$ between typing rules and synthesis rules.
\end{lemma}
\begin{proof}

$\varphi$ can be proven to be a bijection by showing that it is both surjective and injective.
\begin{itemize}
 \item The translation function $\varphi$ maps every typing rule to a synthesis rule, and synthesis rules are constructed exclusively in this manner; thus, $\varphi$ is surjective.
 \item To prove that $\varphi$ is injective, we define its inverse function $\psi$: given a synthesis rule in the form presented in Figure~\ref{S-Rule}, $\psi$ removes all substitutions $\sigma_i$, as well as the unification and assignment acquisition, and restores the terms $\overline{x}_0$ to the conclusion, recovering the typing rule as shown in Figure~\ref{T-Rule}. Therefore, $\psi \circ \varphi = \varphi \circ \psi = id$. The existence of $\psi$ demonstrates that $\varphi$ is injective.
\end{itemize}

\end{proof}

\subsection{Proof for Lemma~\ref{lemma:assignment}}
\label{sec:appendix-lemma2}
\begin{lemma}
For any synthesis judgment derived using a synthesis rule of the form in Figure~\ref{S-Rule}, the composed substitution $\sigma_c = \sigma_n \ldots \sigma_0 \sigma$ constitutes an assignment, in the sense that $\sigma_c(\skhzc{\overline{x}_i})$, $\sigma_c(\skhzc{\overline{x}})$, and $\sigma_c(\skhzc{\overline{y}})$ are all ground terms.
\end{lemma}
\begin{proof}
Induction on the structure of the synthesis derivation.
\begin{itemize}
 \item \textbf{Base case:} The base case occurs when the synthesis rule has no subgoals. In this situation, the free variables $\overline{x}_f$ are precisely those appearing in the typing rule's conclusion, i.e., $\mvhzc(\overline{x}_f) = \mvhzc(\overline{x}_0)$. The unification step produces a substitution $\sigma$ such that $\sigma(\overline{x}_0) = \sigma(\overline{x})$, and the assignment $\sigma_0$ obtained in this step covers all variables in $\sigma(\overline{x}_f)$. 
 
 Consequently, $\sigma_0\sigma(\overline{x}_f)$ are ground terms, thus $\sigma_0\sigma(\overline{x}_0) = \sigma_0\sigma(\overline{x})$ are also ground terms, and so is $\sigma_0\sigma(\overline{y})$ since $\overline{y}$ contains only variables that are in $\overline{x}_0$.
 
 \item \textbf{Inductive case:} The inductive case arises when the synthesis rule has one or more subgoals. By the induction hypothesis, for each $i$, the substitution $\sigma_i$ is an assignment that covers all variables in $\sigma_{i-1}\ldots\sigma_0\sigma(\overline{x}_i)$, ensuring that $\sigma_{i}\ldots\sigma_0\sigma(\overline{x}_i)$ are ground terms, and therefore so is $\sigma_c(\overline{x}_i)$ for all $i\ge1$.
 
 We similarly observe that $\sigma_0\sigma$ covers all variables in $\overline{x}_f$ so $\sigma_c$ must also cover $\overline{x}_f$ and thus $\overline{x}_0$ ($\overline{x}_f$ contains all variables in $\overline{x}_0$ besides those already instantiated in $\overline{x}_i$). 
 We finally have $\sigma_c(\overline{x}) = \sigma_n \ldots \sigma_0 \sigma(\overline{x}) = \sigma_n \ldots \sigma_0 \sigma(\overline{x}_0) = \sigma_c(\overline{x}_0)$ is a ground term, and so is $\sigma_c(\overline{y})$, since $\overline{y}$ contains only variables that are in $\bigcup_{i=0}^n\overline{x}_i$.
 \end{itemize}
\end{proof}

\subsection{Proof for Lemma~\ref{lemma:translation}}
\label{sec:appendix-lemma3}
\begin{lemma}
Let $T_1$ be a well-formed type derivation tree with root node labeled $P(\overline{x}_c)$, annotated with typing rule $R$, and let $T_2$ be its corresponding synthesis derivation tree constructed via the translation, whose root node is labeled with $P(\overline{x}) \leadsto \theta$, where $\theta$ is the assignment synthesized by Algorithm~\ref{alg:synthtree}.
If $\overline{x}$ is unifiable with $\overline{x}_c$, then in the application of the synthesis rule that derives $P(\overline{x}) \leadsto \theta$, (1) the unification step, (2) the acquisition for free variables, and (3) the constraint checking step will always succeed, and (4) the synthesized assignment $\theta$ makes $P(\overline{x})$ concretized to $P(\theta(\overline{x})) = P(\overline{x}_c)$.
\end{lemma}

\begin{proof}
  Induction on the size of the type derivation tree $T_1$. Suppose the typing rule $R$ has the form shown in Figure~\ref{T-Rule}, its corresponding synthesis rule $R'$ has the form shown in Figure~\ref{S-Rule}.
  \begin{itemize}
    \item \textbf{Base case:} 
    The base case is when the type derivation tree $T_1$ has only one node, which means the typing rule $R$ has no premises. 
    In this case, since $\overline{x}$ is unifiable with $\overline{x}_c$ and in the typing rule $R$, $\overline{x}_c$ is the concretization of $\overline{x}_0$, the unification and acquisition steps will succeed, yielding a substitution $\sigma$ such that $\sigma(\overline{x}_0) = \sigma(\overline{x})$ and $\theta'\sigma(\overline{x}_0) = \overline{x}_c$ for some assignment $\theta'$, which can be used to obtain $\sigma_0 = \theta'|_{\mvhzc(\sigma(\overline{x}_f))}= \theta'|_{\mvhzc(\sigma(\overline{x}_0))}$.
    
    By lemma~\ref{lemma:assignment}, $\sigma_0\sigma$ is an assignment that covers all variables in $\overline{x}$, $\overline{x}_0$ and $\overline{y}$, and $\sigma_0\sigma(\overline{x}_0) = \theta'\sigma(\overline{x}_0) = \overline{x}_c$, so if we set $\theta = \sigma_0\sigma|_{\mvhzc(\overline{x})}$ we have $\theta(\overline{x}) = \sigma_0\sigma(\overline{x})= \sigma_0\sigma(\overline{x}_0) = \overline{x}_c$. The constraint checking step is also satisfied because $\mvhzc(y)\subseteq \mvhzc(\overline{x}_0)$ and $\sigma_0\sigma(\overline{y})$ is the same as the verified concretized constraint in the type derivation tree.
    
    \item \textbf{Inductive case:} If the typing rule $R$ has a series premises $P_1(\overline{x}_1), \ldots, P_n(\overline{x}_n)$, then the synthesis rule $R'$ has a series of subgoals $P_1(\meta{\sigma_0\sigma(\overline{x}_1)}) \leadsto \sigma_1, \ldots, P_n(\meta{\sigma_{n-1}\ldots\sigma_0\sigma(\overline{x}_n)}) \leadsto \sigma_n$.
    
    The same as the base case, the unification and acquisition steps will succeed, yielding a substitution $\sigma$ such that $\sigma(\overline{x}_0) = \sigma(\overline{x})$ and $\theta'\sigma(\overline{x}_0) = \overline{x}_c$ for some assignment $\theta'$, and the $\sigma_0 = \theta'|_{\mvhzc(\overline{x}_f)}$.
    
    In the type derivation tree, suppose the $i$-th child node of the root node is labeled with $P_i(\overline{x}_{ic})$, then rule $R$ can be instantiated with some assignment $\theta_c$ such that $\theta_c(\overline{x}_i) = \overline{x}_{ic}$ for all $i$ ($\overline{x}_{0c}$ is just $\overline{x}_{c}$).
    We use the notation $\sigma_i \trianglelefteq \sigma_j$ to denote that $\sigma_j$ factors through $\sigma_i$, meaning that $\sigma_j = \sigma'\sigma_i$ for some substitution $\sigma'$.
    Following the derivation of $\sigma_0$, we know that $\theta'\sigma(\overline{x}_0) = \overline{x}_c$ and $\sigma_0 = \theta'|_{\mvhzc(\sigma(\overline{x}_f))}\trianglelefteq \theta'|_{\mvhzc(\sigma(\overline{x}_0))}$, thus $(\sigma_0\sigma|_{\mvhzc(\overline{x}_0)})\trianglelefteq (\theta'\sigma|_{\mvhzc(\overline{x}_0)}) \trianglelefteq \theta_c$ holds ($\theta_c$ concretizes $\overline{x}_0$) and $\sigma_0\sigma$ ranges over only variables in $\overline{x}_0$ and $\overline{x}$.
    
    Since $\overline{x}_1$ has no common variables with $\overline{x}$, $\sigma_0\sigma(\overline{x}_1)= (\sigma_0\sigma|_{\mvhzc(\overline{x}_0)})(\overline{x}_1)$, and $\theta_c(\overline{x}_1)=\overline{x}_{1c}$, hence $\sigma_0\sigma(\overline{x}_1)$ is unifiable with $\overline{x}_{1c}$.
    Thus we can use the induction hypothesis to synthesis $P_1(\meta{\sigma_0\sigma(\overline{x}_1)}) \leadsto \sigma_1$, where $\sigma_1\sigma_0\sigma(\overline{x}_1) = \overline{x}_{1c}$.
    
    Then we continue that $\sigma_1\sigma_0\sigma(\overline{x}_1) = \theta_c(\overline{x}_1)$, and $(\sigma_1\sigma_0\sigma|_{\mvhzc(\overline{x}_0,\overline{x}_1)}) \trianglelefteq \theta'$ holds, and then $\sigma_1\sigma_0\sigma(\overline{x}_2)=(\sigma_1\sigma_0\sigma|_{\mvhzc(\overline{x}_0,\overline{x}_1)})(\overline{x}_2)$ is unifiable with $\overline{x}_{2c}$, and we can apply the induction hypothesis again and again until all assignments to subgoals are synthesized.
    
    Finally, we can collect all the assignments $\sigma_n, \ldots, \sigma_1, \sigma_0$ and compose them to form $\sigma_c = \sigma_n \ldots \sigma_0\sigma$, which is an assignment and covers all variables in $\overline{x}_i$, $\overline{x}$ and $\overline{y}$ by Lemma~\ref{lemma:assignment}. $\sigma_c(\overline{x}_i)=\overline{x}_{ic}=\theta_c(\overline{x}_i)$ for all $i$, so $\theta_c \trianglelefteq \sigma_c$ since $\sigma_c$ also ranges over $\mvhzc(\overline{x})$.
    We have $\sigma_c(\overline{y}) = \theta_c(\overline{y})$, which instantiates the constraint $\phi(\overline{y})$ to be satisfied as in the type derivation tree. And let $\theta(\overline{x}) = \sigma_c|_{\mvhzc(\overline{x})}$ concretizes the synthesis goal $P(\overline{x})$ to the ground typing judgment $P(\overline{x}_c)$ ($\theta(\overline{x})=\sigma_n \ldots \sigma_0\sigma(\overline{x})=\sigma_n \ldots \sigma_0\sigma(\overline{x}_0)=\sigma_c(\overline{x}_0)=\overline{x}_c$), which is the same as the root node of the type derivation tree.
  \end{itemize}
\end{proof}

\subsection{Proof for Lemma~\ref{lemma:isomorphism2}}
\label{sec:appendix-lemma4}

\begin{lemma}
For any synthesis derivation tree $T_2$, there exists a type derivation tree $T_1$, $T_1 \cong T_2$.
\end{lemma}
\begin{proof}
  Following the translation from synthesis derivation trees to the type derivation trees detailed in \autoref{sec:isomorphism}, the root and child order preservation is easily satisfied. Then we can prove the label correspondence by induction on the size of the synthesis derivation tree $T_2$ (suppose $T_2$'s top-level synthesis rule has the form in \autoref{S-Rule}).

  \begin{itemize}
     \item \textbf{Base case:} The base case occurs when the synthesis derivation tree $T_2$ has only one node, whose synthesis rule has no premises. By the described translation, the node in the type derivation tree $T_1$ is annotated with the typing judgment $P(\theta(\overline{x}))$, if $\theta$ is the synthesized assignment. In this case, the label correspondence is clearly satisfied.
     \item \textbf{Inductive case:} The inductive case arises when the synthesis rule has one or more subgoals. 
     By the induction hypothesis, for the $i$-th subgoal $P_i(\sigma_{i-1}\ldots\sigma_0\sigma(\overline{x}_i)) \leadsto \sigma_i$, we have constructed a sub type derivation tree whose root node is labeled with the typing judgement $P_i(\overline{x}_{ic})$, where $\overline{x}_{ic}=\sigma_{i}\ldots\sigma_0\sigma(\overline{x}_i)$. 
     
     Let $\theta_c=\sigma_{i}\ldots\sigma_0|_{\mvhzc(\overline{x}_0\cup\ldots \overline{x}_n)}$, then we have $\theta_c(\overline{x}_{i})=\sigma_{i}\ldots\sigma_0\sigma(\overline{x}_i)=\overline{x}_{ic}$ for all $i\geq 1$. 
     For the conclusion part, $\theta_c(\overline{x}_{0})=\sigma_0\sigma(\overline{x}_0)=\sigma_0\sigma(\overline{x})=\sigma_n \ldots \sigma_0\sigma(\overline{x})=\theta(\overline{x})$. 
     Thus, under the assignment $\theta_c$, the corresponding typing rule, along with the translated sub type derivation trees, forms a valid type derivation tree $T_1$ whose root node is labeled with $P(\theta(\overline{x}))$.
 \end{itemize}
\end{proof}

\newpage
\section{Additional Information about Benchmarks}
\label{appendix-benchmark}

\subsection{Java Benchmark}
We have carefully designed a subset of the Java programming language that encompasses the fundamental syntactic constructs, including classes, methods, variables, arrays, generic types, operators, and control flow statements. While excluding advanced Java features such as multithreading, lambda expressions, and complex reflection mechanisms.

The Java subset retains sufficient expressiveness to represent core imperative programming paradigms. Our evaluation shows that it covers approximately 78\% of the Java programs in the original MBJP dataset, which is the Java version of MBPP dataset and available at \href{https://github.com/amazon-science/mxeval}{amazon-science/mxeval}. The complete grammar specification is presented below:

\begin{tcolorbox}[
    enhanced,
    breakable,
    colback=blue!5!white,
    colframe=blue!50!black,
    title=Java Subset Language Grammar,
    fonttitle=\bfseries,
    boxrule=1pt,
    arc=3pt
]
\begin{Verbatim}[fontsize=\small]
<program>     := <program> <program>
              |  <modifier> <type> <identifier> ;
              |  <modifier> <type> <identifier> = <term> ;
              |  <modifier> <type> <identifier> ( <stmt> ) { <stmt> }
              |  <identifier> ( <stmt> ) { <stmt> }
              |  <stmt>
              |  class <identifier> { <program> }
<stmt>        := empty
              |  <type> <identifier> = <term>;
              |  <type> <identifier>;
              |  if ( <term> ) { <stmt> }
              |  if ( <term> ) { <stmt> } else { <stmt> }
              |  while ( <term> ) { <stmt> }
              |  do { <stmt> } while ( <term> )
              |  for ( <stmt> ; <term> ; <term> ) { <stmt> }
              |  for ( <type> <identifier> : <term> ) { <stmt> }
              |  return <term>;
              |  continue;
              |  break;
              |  <term>;
              |  switch ( <term> ) { <stmt> }
              |  case <term> : <stmt>
              |  <stmt> <stmt>
<term>        := <identifier>
              |  <integer>
              |  <float>
              |  <character>
              |  <string>
              |  true
              |  false
              |  null
              |  <term> = <term>
              |  ( <type> ) <term>
              |  <term> instanceof <type>
              |  <term> ? <term> : <term>
              |  <term> . <identifier>
              |  <term> [ <term> ]
              |  new <type> ( <term> )
              |  new <type>
              |  new <term> [ <term> ]
              |  <term> . <identifier> ( <term> )
              |  <identifier> ( <term> )
              |  <type>
              |  [ <term_list> ]
              |  { <term_list> }
              |  ( <term> )
              |  <term> + <term>
              |  <term> - <term>
              |  <term> * <term>
              |  <term> / <term>
              |  <term> % <term>
              |  - <term>
              |  <term> << <term>
              |  <term> >> <term>
              |  <term> & <term>
              |  <term> | <term>
              |  <term> ^ <term>
              |  ~ <term>
              |  <term> ++
              |  <term> --
              |  ++ <term>
              |  -- <term>
              |  <term> == <term>
              |  <term> != <term>
              |  <term> < <term>
              |  <term> <= <term>
              |  <term> > <term>
              |  <term> >= <term>
              |  <term> && <term>
              |  <term> || <term>
              |  ! <term>
<type>        := int
              |  float
              |  boolean
              |  char
              |  String
              |  null
              |  <type> []
              |  <identifier> <>
              |  <identifier> < <type> >
              |  <identifier> < <type> , <type> >
              |  <type> -> <type>
              |  <identifier>
              |  void
<term_list>   := <term> , <term_list>
              |  <term>
              |  empty
<modifier>    := public | private | protected | static | final | abstract
<identifier>  := [a-zA-Z_][a-zA-Z0-9_]*
<integer>     := [0-9]+
<float>       := [0-9]+.[0-9]+
<character>   := '[^']'
<string>      := ".*"
\end{Verbatim}
\end{tcolorbox}

\subsection{SuFu Benchmark}

The SuFu dataset is sourced from the GitHub repository accompanying the SuperFusion research paper (\href{https://github.com/jiry17/SuFu}{jiry17/SuFu}). SuperFusion is an advanced program optimization tool that processes SuFu programs as input and applies a series of transformations to improve computational efficiency. 
The SuFu language features a type system that explicitly specifies those intermediate data structures to be optimized, and all code fragments operating on these data structures.
The SuFu programs in this dataset are drawn from two primary lines of prior work. The first source consists of fusion tasks primarily collected from classical textbooks such as Introduction to Algorithms, while the second originates from research on program restructuring, where the goal is to transform a reference program into specific target forms such as divide-and-conquer or single-pass. 
Each program is supplemented with a natural language description. These descriptions were first generated by the GPT-o3-mini model and then carefully reviewed and edited by human experts to correct inaccuracies and improve clarity.

A key issue with the SuFu language is that when an algorithm description is provided, the initial SuFu program still needs to be manually constructed.
However, manually constructing SuFu programs from algorithmic descriptions is highly challenging and requires considerable domain expertise in functional programming. 
Directly prompting language models to generate SuFu code from natural language descriptions often results in programs that are syntactically or semantically flawed, typically containing multiple type errors. 

The following is an example of generated SuFu programs with type errors:
\begin{tcolorbox}[
    colback=gray!5,
    colframe=black,
    boxrule=0.5pt,
    arc=2pt,
    left=8pt,
    right=8pt,
    top=8pt,
    bottom=8pt,
    title=Natural Language Prompt
]
\textit{Write a function for the mts algorithm, which finds the maximum sum among all possible tail lists of the input list.}
\end{tcolorbox}
\begin{tcolorbox}[
    colback=gray!5,
    colframe=black,
    boxrule=0.5pt,
    arc=2pt
]
\begin{Verbatim}[fontsize=\small, xleftmargin=1em]
tails :: List -> Packed NList
tails Nil = let ts = NCons(Nil, NNil) in
            rewrite (label ts)
tails Cons(_, t)@xs = let ts = tails t in
                      rewrite (label NCons(xs, ts))
                      // there is a type error, `ts` should be `unlabel ts`
mts_label xs = let ts = tails xs in
                rewrite (maximum (map sum (unlabel ts)))
\end{Verbatim}
\end{tcolorbox}

The complete SuFu grammar specification is as follows:
\begin{tcolorbox}[
    colback=blue!5!white,
    colframe=blue!50!black,
    title=SuFu Language Grammar,
    fonttitle=\bfseries,
    boxrule=1pt,
    arc=3pt,
    breakable
]
\begin{Verbatim}[fontsize=\small]
<program>     := <program> <program>
              |  <command>
<command>     := Inductive <Name> = <name> <type> 
                                    | ... 
                                    | <name> <type>;
              |  <Name> = <type>;
              |  <name> = <term>;
<type>        := {<type>, ..., <type>}
              |  <type> -> <type>
              |  Reframe <type>
              |  Unit 
              |  Int 
              |  Bool 
              |  <Name>
<term>        := unit 
              |  <integer> 
              |  true 
              |  false 
              |  <name>
              |  + | - | * | / 
              |  < | <= | > | >= | == 
              |  and | or | not
              |  <term> <term>
              |  \<name>:<type>. <term>
              |  fix <term>
              |  let <name> = <term> in <term>
              |  if <term> then <term> else <term>
              |  {<term>, ..., <term>}
              |  <term>.<integer>
              |  match <term> with <pattern> -> <term> 
                                   | ... 
                                   | <pattern> -> <term> 
                 end
              |  rewrite <term> 
              |  label <term> 
              |  unlabel <term>
<pattern>     := _ 
              |  <name>
              |  <name> <pattern>
              |  {<pattern>, ..., <pattern>}
<Name>        := [A-Z][a-zA-Z0-9_]*
<name>        := [a-z][a-zA-Z0-9_]*
<integer>     := [0-9]+
\end{Verbatim}
\end{tcolorbox}
\end{document}